\pdfoutput=1
\RequirePackage{ifpdf}
\ifpdf 
\documentclass[pdftex]{sigma}
\else
\documentclass{sigma}
\fi

\numberwithin{equation}{section}

\newtheorem{Theorem}{Theorem}[section]
\newtheorem*{Theorem*}{Theorem}
\newtheorem{Corollary}[Theorem]{Corollary}

\newtheorem{Proposition}[Theorem]{Proposition}
 { \theoremstyle{definition}
\newtheorem{Definition}[Theorem]{Definition}

}

\begin{document}
\allowdisplaybreaks

\newcommand{\arXivNumber}{2210.14180}

\renewcommand{\PaperNumber}{025}

\FirstPageHeading

\ShortArticleName{The $B_{2}$ Harmonic Oscillator with Reflections and Superintegrability}

\ArticleName{The $\boldsymbol{B_{2}}$ Harmonic Oscillator with Reflections\\ and Superintegrability}

\Author{Charles F.~DUNKL}

\AuthorNameForHeading{C.F.~Dunkl}

\Address{Department of Mathematics, University of Virginia,\\ PO Box 400137, Charlottesville VA 22904-4137, USA}
\Email{\href{mailto:cfd5z@virginia.edu}{cfd5z@virginia.edu}}
\URLaddress{\url{https://uva.theopenscholar.com/charles-dunkl}}

\ArticleDates{Received October 27, 2022, in final form April 17, 2023; Published online April 25, 2023}

\Abstract{The two-dimensional quantum harmonic oscillator is modified with reflection terms associated with the action of the Coxeter group $B_{2}$, which is the symmetry group of the square. The angular momentum operator is also modified with reflections. The~wave\-functions are known to be built up from Jacobi and Laguerre polynomials. This paper introduces a fourth-order differential-difference operator commuting with the Hamiltonian but not with the angular momentum operator; a specific instance of superintegrability. The action of the operator on the usual orthogonal basis of wavefunctions is explicitly described. The wavefunctions are classified according to the representations of the group: four of degree one and one of degree two. The identity representation encompasses the wavefunctions invariant under the group. The paper begins with a short discussion of the modified Hamiltonians associated to finite reflection groups, and related raising and lowering operators. In particular, the Hamiltonian for the symmetric groups describes the Calogero--Sutherland model of identical particles on the line with harmonic confinement.}


\Classification{81R12; 37J35; 33C45; 81Q05}

\section{Introduction}

The two-dimensional quantum harmonic oscillator is modified with reflection
terms associated with the action of the Coxeter group $B_{2}$, the symmetry
group of the square. The wavefunctions are known to be built up from Jacobi
and Laguerre polynomials. This paper introduces a fourth-order
differential-difference operator commuting with the Hamiltonian but not with
the angular momentum operator; a specific instance of superintegrability. The
action of the operator on an orthogonal basis of wavefunctions is explicitly
described. The wavefunctions are not in general invariant under the group,
rather are classified by the representations of the group: four of degree one
and one of degree two. The group-invariant wavefunctions of the $B_{2}$
oscillator and its superintegrability have been studied by Tremblay et al.~\cite{TremblayETal2009,TremblayETal2010}, Quesne~\cite{Quesne2010}.

First the general background on finite reflection groups and root systems,
Dunkl operators, and the associated Hamiltonian is described. In particular, the
Calogero--Sutherland model of~$N$ identical particles on a line with $r^{-2}$
interaction and harmonic confinement comes from the symmetric group (Lassalle~\cite{Lassalle1991}, Baker and Forrester~\cite{BakerForrester1997}). In the
general situation, there are raising and lowering operators which can be used
to construct operators commuting with the Hamiltonian and the group action.
After this the development turns to dihedral groups (type~$I_{2}(k)$) and the use of a complex coordinate system, which simplifies the
description of rotations. Some general formulas are specialized to this setting.

The description of the wavefunctions of $B_{2}$ and of the action of specific
operators is in Sections~\ref{WFbasis} and~\ref{ExpH0}. The important
operators are the Hamiltonian $\mathcal{H}$, the angular momentum
$\mathcal{J}$ and a new operator $\mathcal{K}$ which commutes with
$\mathcal{H}$ and the group action but not with $\mathcal{J}^{2}$. This
property constitutes superintegrability. There are a number of different
classes of wavefunctions, requiring frequent case-by-case analysis. The
explicit action of $\mathcal{K}$ on the orthogonal basis of wavefunctions is
found in Section~\ref{ExpK}.

In the appendix, there are details on some proofs, and a sketch of a symbolic
computation method of proving relations involving polynomials and Dunkl operators.

\section{Reflection groups and a harmonic oscillator}

In $\mathbb{R}^{N}$ the inner product is $\langle x,y\rangle
:=\sum_{i=1}^{N}x_{i}y_{i}$ and $\Vert x\Vert^{2}=\langle
x,x\rangle $. If $v\neq0$, then the reflection~$\sigma_{v}$ along $v$ is
defined by%
\[
x\sigma_{v}:=x-2\frac{\langle x,v\rangle}{\Vert v\Vert^{2}}v.
\]
This is an isometry $\Vert x\sigma_{v}\Vert^{2}=\Vert
x\Vert^{2}$ and an involution $\sigma_{v}^{2}=I$. The set of fixed
points ($x\sigma_{v}=x$) is the hyperplane $\{ x\colon \langle
x,v\rangle =0\} $. A finite root system is a subset $R$ of
nonzero elements of $\mathbb{R}^{N}$ satisfying $u,v\in R$ implies
$u\sigma_{v}\in R$. We restrict consideration to reduced root systems, that is
if $u,cu\in R$, then $c=\pm1$. Define $W(R) $ to be the group
generated by $\{ \sigma_{v}\colon v\in R\} $; this is a finite subgroup
of the orthogonal group $O_{N}(\mathbb{R})$. There is a
decomposition of $R$ into $R_{+}$ (the positive roots) and $R_{-}$; this
relies on choice of a vector $u$ such that $\langle u,v\rangle
\neq0$ for all $v\in R$ then set $R_{+}=\{ v\in R\colon \langle
u,v\rangle >0\} $. Since $\sigma_{v}=\sigma_{-v}$, the set $R_{+}$
can be used to index the reflections in $W(R) $. The set of
reflections $\sigma_{v}$ decomposes into conjugacy classes $(W$ orbits)
$\sigma_{u}\sim\sigma_{v}$ if $u=vw$ for some $w\in W(R)$. A
\textit{multiplicity function} $\kappa_{v}$ is a function on $R$ which is
constant on each conjugacy class, usually here $\kappa_{v}\geq1$. Set
$\gamma_{\kappa}:=\sum_{v\in R_{+}}\kappa_{v}$. Define the Dunkl operator
($1\leq i\leq N$)%
\[
\mathcal{D}_{i}f(x) :=\frac{\partial}{\partial x_{i}}f(x)
+\sum_{v\in R_{+}}\kappa_{v}\frac{f(x)-f(x\sigma_{v})}{\langle x,v\rangle}v_{i}.
\]
Then $\mathcal{D}_{i}\mathcal{D}_{j}=\mathcal{D}_{j}\mathcal{D}_{i}$ for all
$i$, $j$ (Dunkl \cite{Dunkl1989}, also see Dunkl and Xu \cite[Theorem~6.4.8]%
{DunklXu2014}). Let
\[
\nabla=\bigg(\frac{\partial}{\partial x_{1}},\ldots,\frac{\partial}{\partial x_{N}}\bigg), \qquad \Delta=\sum_{i=1}^{N}\bigg(\frac{\partial}{\partial x_{i}}\bigg)^{2}\qquad \text{and }\qquad \nabla_{\kappa}=(\mathcal{D}_{1},\ldots,\mathcal{D}_{N}).
\]

The Dunkl Laplacian is $\Delta_{\kappa}:=\sum_{i=1}^{N}\mathcal{D}_{i}^{2}$
and
\[
\Delta_{\kappa}f(x) =\Delta f(x) +\sum_{v\in R_{+}}\kappa_{v}\bigg(2\frac{\langle \nabla f(x)
,v\rangle}{\langle x,v\rangle}-\Vert v\Vert^{2}\frac{f(x) -f(x\sigma_{v})}{\langle
x,v\rangle^{2}}\bigg).
\]
This leads to the modified Schr\"{o}dinger equation (with parameter $\omega>0$)%
\[
\mathcal{H}\psi:=\big(\omega^{2}\Vert x\Vert^{2}-\Delta_{\kappa}\big) \psi=E\psi.
\]
The exponential ground state is $g(x) :=\exp\bigl(-\frac{\omega}{2}\Vert x\Vert^{2}\bigr)$, as can be seen from
the transformation%
\[
g^{-1}\big(\omega^{2}\Vert x\Vert^{2}-\Delta_{\kappa}\big)
(fg) =-\Delta_{\kappa}f+\omega(N+2\gamma_{\kappa
}+2\langle x,\nabla\rangle ) f,
\]
which implies $\big(\omega^{2}\Vert x\Vert^{2}-\Delta_{\kappa
}\big) g=\omega(N+2\gamma_{\kappa}) g$. An equivalent
expression is%
\begin{equation}
g^{-1}\mathcal{H}g=-\Delta_{\kappa}+\omega\sum_{i=1}^{N}(x_{i}%
\mathcal{D}_{i}+\mathcal{D}_{i}x_{i}). \label{gHgD}%
\end{equation}
Denote the set of polynomials on $\mathbb{R}^{N}$ by $\mathcal{P}$ and the set
of polynomials homogeneous of degree~$n$ by~$\mathcal{P}_{n}$ (that is,
$p(cx) =c^{n}p(x)$ for $c\in\mathbb{R}$). Let
$\mathcal{H}_{\kappa,n}=\{ p\in\mathcal{P}_{n}\colon \Delta_{\kappa
}p=0\} $ ($\kappa$-\textit{harmonic} polynomials). We find
eigenfunctions of $g^{-1}\mathcal{H}g$ of the form $p(x)
q\big(\omega\Vert x\Vert^{2}\big) $ with $p\in
\mathcal{H}_{\kappa,n}$ (thus $\Delta_{\kappa}p=0$ and $\langle
x,\nabla\rangle p=np$). This gives the differential equation (where
$t=\omega\Vert x\Vert^{2}$)
\[
t\frac{{\rm d}^{2}}{{\rm d}t^{2}}q+\bigg(n+\frac{N}{2}+\gamma_{\kappa}-t\bigg)
\frac{\rm d}{{\rm d}t}q-\frac{1}{4}\bigg(2n+N+2\gamma_{\kappa}-\frac{E}{\omega}\bigg) q=0
\]
and the solution is the Laguerre polynomial $q(t) =L_{m}^{(\alpha)}(t)$, $m=0,1,2,\ldots$,
$\alpha=\gamma_{\kappa}+n+\frac{N}{2}-1$, $E=\omega(N+2\gamma_{\kappa}+2n+4m)$. Note $E$ depends on $\deg(pq) =n+2m$. The Laguerre polynomial of degree $n$ and index $\alpha>-1$ satisfies
\begin{align*}
&L_{n}^{(\alpha)}(t) :=\frac{(\alpha+1)_{n}}{n!}\sum_{j=0}^{n}\frac{(-n)_{j}}{(\alpha+1)_{j}}\frac{t^{j}}{j!},
\\
&\int_{0}^{\infty}L_{m}^{(\alpha)}(t)L_{k}^{(\alpha)}(t) t^{\alpha}{\rm e}^{-t}\,\mathrm{d}t
 =\delta_{mk}\frac{\Gamma(\alpha+1+m)}{m!}=\delta_{mk}
\Gamma(\alpha+1) \frac{(\alpha+1)_{m}}{m!}.
\end{align*}
The Pochhammer symbol is $(a)_{n}=\prod_{i=1}^{n}(a+i-1)$ (or $(a)_{0}=1$ and $(a)_{n+1}=(a+n)(a)_{n}$).

There is an orthogonality structure which uses the $W(R)$-invariant weight function
\[
h_{\kappa}(x) :=\prod_{v\in R_{+}}\vert \langle x,v\rangle \vert^{\kappa_{v}}
\]
positively homogeneous of degree $\gamma_{\kappa}$. The orthogonality
$\mathcal{H}_{\kappa,n}\bot\mathcal{H}_{\kappa,m}$ for $n\neq m$ holds with
respect to the measure $h_{\kappa}(x)^{2}{\rm d}\mu(x)$ on the sphere
$S^{N-1}:=\{x\colon \Vert x\Vert =1\}$,
where $\mu$ is the rotation-invariant surface measure. There is a key result
on adjoints: suppose $p,q$ are sufficiently smooth and have exponential decay
then (with $1\leq i\leq N$)
\begin{equation}
\int_{\mathbb{R}^{N}}(\mathcal{D}_{i}p) qh_{\kappa}^{2}\,\mathrm{d}m
=-\int_{\mathbb{R}^{N}}p(\mathcal{D}_{i}q)h_{\kappa}^{2}\,\mathrm{d}m,
\label{D*D}
\end{equation}
where $\mathrm{d}m$ is Lebesgue measure on $\mathbb{R}^{N}$ (see~\cite[Theorem~7.7.10]{DunklXu2014}). Thus the adjoint of $\mathcal{D}_{i}$ is defined on a
dense subspace of $L^{2}\big(\mathbb{R}^{N},h_{\kappa}^{2}\mathrm{d}
m\big)$ and $\mathcal{D}_{i}^{\ast}=-\mathcal{D}_{i}$. This meaning of
adjoint will be used throughout. Furthermore, the conjugate of $\mathcal{H}$ is
\[
h_{\kappa}\bigl(-\Delta_{\kappa}+\omega^{2}\Vert x\Vert
^{2}\bigr) h_{\kappa}^{-1}=-\Delta+\omega^{2}\Vert x\Vert
^{2}+\sum_{v\in R_{+}}\frac{\kappa_{v}(\kappa_{v}-\sigma_{v})
\Vert v\Vert^{2}}{\langle x,v\rangle^{2}},
\]
(details of the derivation are in Appendix~\ref{hHh}) a Schr\"{o}dinger
equation with the potential%
\[
V(x)=\omega^{2}\Vert x\Vert^{2}+\sum_{v\in R_{+}}\frac{\kappa
_{v}(\kappa_{v}-\sigma_{v}) \Vert v\Vert^{2}}{\langle x,v\rangle^{2}},
\]
which includes reflections. The ground state is $h_{\kappa}g$. For the special
case where $R$ is the root system of type $A_{N-1}$ and $W(R)
=\mathcal{S}_{N}$ (the symmetric group), this potential occurs in the
Calogero--Sutherland model of $N$ identical particles on a line with $r^{-2}$
interaction potential and harmonic confinement. There is a closely related
model of $N$ identical particles on a circle with~$r^{-2}$ interaction, called
the trigonometric model. The wavefunctions are Jack polynomials in the
variables $x_{j}={\rm e}^{\mathrm{i}\theta_{j}}$, $1\leq j\leq N$. Lapointe and Vinet
\cite{LapointeVinet1996} defined raising and lowering operators and found
Rodrigues formulas for the Jack polynomials arising in this model. The Jack
polynomials can be used as bases for generalized Hermite (Lassalle
\cite{Lassalle1991}) and Laguerre polynomials, which occur as wavefunctions in
types $A$ and $B$ models on the line (Baker and Forrester
\cite{BakerForrester1997}, also see \cite[Section~11.6.3]{DunklXu2014}).

We need the basic commutation relations ($[A,B] :=AB-BA$) for
$a,b\in\mathbb{R}^{N}$:%
\begin{align}
\label{xDx}
&[\langle a,\nabla_{\kappa}\rangle,\langle b,x\rangle]
=\langle a,b\rangle +2\sum_{v\in R_{+}}\kappa_{v}
\frac{\langle a,v\rangle \langle b,v\rangle}{\Vert v\Vert^{2}}\sigma_{v},
\\
\label{Dbxsq}
&[\Delta_{\kappa},\langle b,x\rangle] =2\langle b,\nabla_{\kappa}\rangle,\qquad
\big[\Vert x\Vert^{2},\langle a,\nabla_{\kappa}\rangle\big] =-2\langle a,x\rangle.
\end{align}

\begin{Definition}
For $a,b\in\mathbb{R}^{N}$, the angular momentum operator is $J_{a,b}
:=\langle a,x\rangle \langle b,\nabla_{\kappa}\rangle
-\langle b,x\rangle \langle a,\nabla_{\kappa}\rangle$.
\end{Definition}

\begin{Proposition}
\label{Jprops}
$J_{a,b}=\langle b,\nabla_{\kappa}\rangle
\langle a,x\rangle -\langle a,\nabla_{\kappa}\rangle
\langle b,x\rangle $; $J_{a,b}^{\ast}=-J_{a,b}$ and $[\mathcal{H},J_{a,b}] =0$.
\end{Proposition}

\begin{proof}
From \eqref{xDx}, the commutator $[\langle a,x\rangle,\langle b,\nabla_{\kappa}\rangle]
=-[\langle b,x\rangle,\langle a,\nabla_{\kappa}\rangle]$. This
proves the first statement. By \eqref{D*D}, $J_{a,b}^{\ast}=-\langle
a,x\rangle \langle b,\nabla\kappa\rangle +\langle
b,x\rangle \langle a,\nabla_{\kappa}\rangle =-J_{a,b}$. Next
by~\eqref{Dbxsq},
\begin{align*}
[\mathcal{H},\langle b,\nabla_{\kappa}\rangle \langle a,x\rangle ]
& =\omega^{2}\big\{ \Vert x\Vert^{2}\langle b,\nabla_{\kappa}\rangle \langle a,x\rangle
-\langle b,\nabla_{\kappa}\rangle \Vert x\Vert^{2}\langle a,x\rangle \big\}
\\
&\phantom{=} -\{ \Delta_{\kappa}\langle b,\nabla_{\kappa}\rangle
\langle a,x\rangle -\langle b,\nabla_{\kappa}\rangle
\langle a,x\rangle \Delta_{\kappa}\}
\\
& =\omega^{2}\big[ \Vert x\Vert^{2},\langle b,\nabla_{\kappa}\rangle \big] \langle a,x\rangle
-\langle b,\nabla_{\kappa}\rangle [ \Delta_{\kappa},\langle a,x\rangle]
\\
& =-2\omega^{2}\langle b,x\rangle \langle a,x\rangle
-2\langle b,\Delta_{\kappa}\rangle \langle a,\nabla_{\kappa}\rangle,
\end{align*}
and this expression is symmetric in $a,b$ and thus $[\mathcal{H},J_{a,b}] =0$.
\end{proof}

\begin{Corollary}
$[ \Delta_{\kappa},J_{a,b}] =0$ and $\big[\Vert x\Vert^{2},J_{a,b}\big] =0$.
\end{Corollary}

This family of angular momentum operators has been studied by Feigin and
Hakobyan \cite{FeiginHakobyan2015}, especially in connection with the
symmetric group and the Calogero--Moser model.

We introduce raising and lowering operators. These operators were used by
Feigin \cite{Feigin2012} in his study of generalized Calogero--Moser models,
which are constructed in terms of subdiagrams (certain subsets of roots) of
the Coxeter diagram of $W(R)$. Note $\{A,B\}:=AB+BA$.

\begin{Definition}
For $a\in\mathbb{R}^{N}$, $a\neq0$, let $A_{a}^{\pm}=\omega\langle
a,x\rangle \pm\langle a,\nabla_{\kappa}\rangle $ and
$H_{a}:=\frac{1}{2}\{ A_{a}^{+},A_{a}^{-}\}
=\omega^{2}\langle a,x\rangle^{2}-\langle a,\nabla_{\kappa}\rangle^{2}$.
\end{Definition}

\begin{Proposition}
\label{[hH'0}
$(A_{a}^{+})^{\ast}=A_{a}^{-}$; $g^{-1}A_{a}^{+}g=\langle a,\nabla_{\kappa}\rangle$ $($lowering$)$ and
$g^{-1}A_{a}^{-}g=2\omega\langle a,x\rangle -\langle a,\nabla_{\kappa}\rangle$ $($raising$)$; $H_{a}^{\ast}=H_{a}$ and $[\mathcal{H},H_{a}] =0$. Also $g^{-1}H_{a}g=\omega(\langle a,x\rangle \langle a,\nabla_{\kappa}\rangle+\langle a,\nabla_{\kappa}\rangle \langle a,x\rangle
) -\langle a,\nabla_{\kappa}\rangle^{2}$.
\end{Proposition}

\begin{proof}
From \eqref{D*D}, it follows that $(A_{a}^{+})^{\ast}=A_{a}^{-}$
and $H_{a}^{\ast}=H_{a}$. The commutator
\[
\big[\omega^{2}\langle a,x\rangle^{2}-\langle a,\nabla_{\kappa}\rangle^{2}, \mathcal{H}\big]
=-\big[\langle a,\nabla_{\kappa}\rangle^{2},\omega^{2}\Vert x\Vert^{2}\big]
-\big[\omega^{2}\langle a,x\rangle^{2},\Delta_{\kappa}\big]
\]
and expanding the right hand side with formulas~\eqref{Dbxsq}
and $\big[A^{2},B\big] =A[ A,B] +[ A,B] A$ shows ${[H_{a},\mathcal{H}] =0}$.
\end{proof}

\begin{Proposition}
\label{wHaw}
Suppose $w\in W(R)$, then $w^{-1}H_{a}w=H_{aw}$;
suppose $S\subset R_{+}$ and $S\cup(-S)$ is an $W(R) $-orbit $($closed under $v\rightarrow vw)$, then $\sum_{v\in S}H_{v}^{k}$ commutes with each $w\in W(R)$ for $k=1,2,3,\ldots$.
\end{Proposition}

\begin{proof}
This follows from $\langle a,\nabla_{\kappa}\rangle w=w\langle
aw,\nabla_{\kappa}\rangle$ (see~\cite[Proposition~6.4.3]{DunklXu2014}) and
\[
w(\langle aw,x\rangle p(x))=\langle aw,xw\rangle p(xw) =\langle
a,x\rangle wp(x)
\]
(because $w\in O_{N}(\mathbb{R})$).
\end{proof}

This produces a collection of self-adjoint operators commuting with $W(R) $ and $\mathcal{H}$.

\section{The dihedral groups}

For $m=3,4,\ldots$, the dihedral group $I_{2}(m)$ is the symmetry
group of the regular $m$-gon. We will use complex coordinates for $\mathbb{R}^{2}$:%
\[
z:=x_{1}+\mathrm{i}x_{2},\qquad
\overline{z}:=x_{1}-\mathrm{i}x_{2}.
\]
Let $\zeta:=\exp\big(\frac{2\pi\mathrm{i}}{m}\big) $, then the
reflections in $I_{2}(m) $ are $\sigma_{j}\colon (z,\overline{z})\rightarrow\big(\overline{z}\zeta^{j},z\zeta
^{-j}\big) $ ($0\leq j<m$), and the rotations are $\rho_{j}\colon (z,\overline{z}) \rightarrow\big(z\zeta^{j},\overline{z}\zeta
^{-j}\big) $. Then $\sigma_{k}\sigma_{j}\sigma_{k}=\sigma_{2k-j}$ and
$\rho_{k}^{-1}\sigma_{j}\rho_{k}=\sigma_{j+2k}$; when~$m$ is even, there are
two conjugacy classes $\{\sigma_{2j}\} $ and $\{\sigma_{2j+1}\}$ with
$0\leq j\leq\frac{m}{2}-1$. The real root vector
for $\sigma_{j}$ is $v_{j}:=\big(\sin\big(\frac{\pi j}{m}\big),-\cos\big(\frac{\pi j}{m}\big)\big)$
and $\langle x,v_{j}\rangle =\frac{\mathrm{i}}{2}\exp\bigl(-\frac{j\pi\mathrm{i}
}{m}\big) \big(z-\zeta^{j}\overline{z}\big)$. The Dunkl operators
are $\big(\partial_{z}:=\frac{\partial}{\partial z},\partial_{\overline{z}}
:=\frac{\partial}{\partial\overline{z}}\big)$
\begin{align*}
&Tf(z) =\partial_{z}f(z) +\sum_{j=0}%
^{m-1}\kappa_{j}\frac{f(z) -f\big(\overline{z}\zeta
^{j}\big)}{z-\overline{z}\zeta^{j}}=\frac{1}{2}(\mathcal{D}%
_{1}-\mathrm{i}\mathcal{D}_{2}) f,\\
&\overline{T}f(z) =\partial_{\overline{z}}f(z)
-\sum_{j=0}^{m-1}\kappa_{j}\frac{f(z) -f\big(\overline
{z}\zeta^{j}\big)}{z-\overline{z}\zeta^{j}}\zeta^{j}
=\frac{1}{2}(\mathcal{D}_{1}+\mathrm{i}\mathcal{D}_{2}) f.
\end{align*}
These imply $\Delta_{\kappa}=4T\overline{T}$. If $m$ is odd, then $\kappa
_{j}=\kappa$; if $m$ is even then $\kappa_{2j}=\kappa_{0}$ and $\kappa
_{2j+1}=\kappa_{1}$ for all $j$. Denote $H_{v_{j}}$ by $H_{j}$ and let
$\widehat{H}_{j}=g^{-1}H_{j}g$. In the complex coordinates,
\[
\langle v_{j},x\rangle =\frac{\mathrm{i}}{2}
\exp\biggl(-\frac{j\pi\mathrm{i}}{m}\biggr) \big(z-\zeta^{j}\overline{z}\big),\qquad
\langle v_{j},\nabla_{\kappa}\rangle =-\mathrm{i}\biggl(\exp\frac{j\pi\mathrm{i}}{m}\biggr) \big(T-\zeta^{-j}\overline{T}\big),
\]
\big(note $\big(\exp\frac{j\pi\mathrm{i}}{m}\big)^{2}=\zeta^{j}$\big), thus
\[
\langle v_{j},x\rangle \langle v_{j},\nabla_{\kappa}\rangle
=\frac{1}{2}\big(z-\zeta^{j}\overline{z}\big) \big(T-\zeta^{-j}\overline{T}\big)
\]
and
\[
\langle v_{j},\nabla_{\kappa}\rangle^{2}=-\zeta^{j}\big({T-\zeta^{-j}\overline{T}}\big)
^{2} =-\zeta^{j}T^{2}+2T\overline{T}-\zeta^{-j}\overline{T}^{2}.
\]
Then using \eqref{gHgD},
\begin{align}
&\widehat{H}_{j} =\zeta^{j}T^{2}-2T\overline{T}+\zeta^{-j}\overline{T}^{2}
 +\frac{\omega}{2}\big\{ \big(z-\zeta^{j}\overline{z}\big) \big(T-\zeta^{-j}\overline{T}\big) +\big(T-\zeta^{-j}\overline{T}\big)
\big(z-\zeta^{j}\overline{z}\big) \big\},\nonumber
\\
&g^{-1}\mathcal{H}g =-4T\overline{T}+\omega\big\{ zT+\overline
{z}\overline{T}+Tz+\overline{T}\overline{z}\big\}.\label{gHg}
\end{align}
In $\mathbb{R}^{2}$ there is only one angular momentum operator (up to scalar
multiplication), namely $x_{1}\mathcal{D}_{2}-x_{2}\mathcal{D}_{1}%
=\mathrm{i}\big(zT-\overline{z}\overline{T}\big) $. Set $\mathcal{J}%
:=zT-\overline{z}\overline{T}$.

The $\kappa$-harmonic polynomials can be found in \cite[Section~7.6]%
{DunklXu2014}; they are expressed in terms of Gegenbauer, respectively Jacobi,
polynomials, in case of odd $m$, respectively even $m$.

\section[Orthogonal basis of wavefunctions for B\_2]{Orthogonal basis of wavefunctions for $\boldsymbol{B_{2}}$}\label{WFbasis}

Henceforth, we specialize to the group $B_{2}=I_{2}(4) $. The
formulas in the previous section apply with $\zeta=\mathrm{i}$. Let
$\gamma_{\kappa}=2\kappa_{0}+2\kappa_{1}$. The weight function $h_{\kappa
}=\big\vert z^{2}-\overline{z}^{2}\big\vert^{\kappa_{0}}\big\vert
z^{2}+\overline{z}^{2}\big\vert^{\kappa_{1}}$. The group has five
irreducible representations: four of degree one and one of degree two. The
four multiplicative characters satisfy $\chi_{0}(\sigma_{k})
=1$, $\chi_{1}(\sigma_{k}) =(-1)^{k}$, $\chi
_{2}(\sigma_{k}) =(-1)^{k+1}$, $\chi_{3}(\sigma_{k}) =-1$, $0\leq k\leq3$. The basis of wavefunctions
(solutions of $\mathcal{H}\psi=2\omega(n+1+\gamma_{\kappa})
\psi$) are denoted $\psi_{n-j,j}$ where the subscript refers to a dominant
monomial in the polynomial part (ignoring $g$) and monomials $z^{a}%
\overline{z}^{b}$ ($a+b=n$) are ordered by $\vert a-b\vert $. The
factor $g$ in the wavefunctions will be omitted and we use operators in the
form $g^{-1}Ag$ acting on polynomials.

The basis functions are all expressed in the following:%
\begin{align*}
R_{n}^{(\alpha,\beta)}(z) :={}&(z\overline{z})^{2n}P_{n}^{(\alpha,\beta)}
\bigg(\frac{z^{4}+\overline{z}^{4}}{2z^{2}\overline{z}^{2}}\bigg)
\\
={}&\frac{(-1)^{n}}{2^{2n}n!}\sum_{j=0}^{n}\binom{n}{j}(-n-\alpha)_{n-j}( -n-\beta)_{j}\big(z^{2}%
-\overline{z}^{2}\big)^{2j}\big(z^{2}+\overline{z}^{2}\big)^{2n-2j}.
\end{align*}
The Jacobi polynomial $P_{n}^{(\alpha,\beta)}(t)$ of degree $n$ and indices $\alpha$, $\beta$ can be defined as (see \cite[Proposition~4.14]{DunklXu2014})
\[
P_{n}^{(\alpha,\beta)}(t) :=\frac{(\alpha+1)_{n}}{n!}\bigg(\frac{1+t}{2}\bigg)^{n}\sum_{j=0}%
^{n}\frac{(-n)_{j}(-n-\beta)_{j}}{(\alpha+1)_{j}~j!}\bigg(\frac{t-1}{t+1}\bigg)^{j};
\]
this formula leads to the expression stated above. Then define%
\begin{gather*}
p_{4n,00}(z) :=R_{n}^{(\kappa_{0}-1/2,\kappa_{1}-1/2)}(z),
\\
p_{4n,11}(z) :=\big(z^{4}-\overline{z}^{4}\big)R_{n-1}^{(\kappa_{0}+1/2,\kappa_{1}+1/2)}(z),
\\
p_{4n+2,10}(z) :=\big(z^{2}+\overline{z}^{2}\big)R_{n}^{(\kappa_{0}-1/2,\kappa_{1}+1/2)}(z),
\\
p_{4n+2,01}(z) :=\big(z^{2}-\overline{z}^{2}\big)R_{n}^{(\kappa_{0}+1/2,\kappa_{1}-1/2)}(z).
\end{gather*}
These are of isotype $\chi_{0}$, $\chi_{3}$, $\chi_{1}$, $\chi_{2}$, respectively. The
$L^{2}$-norms are necessary for normalization, and are derived from%
\begin{gather*}
\int_{0}^{\pi/2}\sin^{2\alpha}\theta\cos^{2\beta}\theta P_{n}^{\left(\alpha-\frac{1}{2},\beta-\frac{1}{2}\right)}(\cos2\theta)
P_{k}^{\left(\alpha-\frac{1}{2},\beta-\frac{1}{2}\right)}(\cos2\theta)\, \mathrm{d}\theta
\\ \qquad
{}=\frac{1}{2}\delta_{nk}B\bigg(\alpha+\frac{1}{2},\beta+\frac{1}{2}\bigg)
\frac{\big(\alpha+\frac{1}{2}\big)_{n}\big(\beta+\frac{1}{2}\big)_{n}(\alpha+\beta+n)}{n!(\alpha
+\beta+1)_{n}(\alpha+\beta+2n)}.
\end{gather*}
The beta function $B$ is defined by a definite integral and satisfies
$B(a,b) =\Gamma(a) \Gamma(b)/\Gamma(a+b) $. Denote for polynomials $p(z,\overline{z})$%
\begin{align*}
&\Vert p\Vert_{\mathbb{T}}^{2} :=\int_{-\pi}^{\pi}\big\vert
p\big({\rm e}^{\mathrm{i}\theta}\big) \big\vert^{2}h_{\kappa}\big({\rm e}^{\mathrm{i}\theta}\big)^{2}\,\mathrm{d}\theta,
\\
&\Vert p\Vert^{2} :=\int_{0}^{\infty}\exp\bigl(-\omega
r^{2}\bigr) r^{2\gamma_{\kappa}+1}\mathrm{d}r\int_{-\pi}^{\pi}\big\vert
p\big(r{\rm e}^{\mathrm{i}\theta}\big) \big\vert^{2}h_{\kappa}\big({\rm e}^{\mathrm{i}\theta}\big)^{2}\,\mathrm{d}\theta,
\end{align*}
then%
\[
\Vert 1\Vert_{\mathbb{T}}^{2} =2^{\gamma_{\kappa}+1}B\bigg(\kappa_{0}+\frac{1}{2},\kappa_{1}+\frac{1}{2}\bigg).
\]
The squared norms are%
\begin{gather*}
\begin{split}
&\Vert p_{4n,00}\Vert_{\mathbb{T}}^{2} =\frac{\big(\kappa_{0}+\frac{1}{2}\big)_{n}\big(\kappa_{1}+\frac{1}{2}\big)_{n} (\kappa_{0}+\kappa_{1}+n)}{n!(\kappa_{0}+\kappa_{1}+1)_{n}
(\kappa_{0}+\kappa_{1}+2n)}\Vert 1\Vert_{\mathbb{T}}^{2},
\\
&\Vert p_{4n,11}\Vert_{\mathbb{T}}^{2} =16\frac{\big(\kappa_{0}+\frac{1}{2}\big)_{n}\big(\kappa_{1}+\frac{1}{2}\big)_{n} (\kappa_{0}+\kappa_{1}+n+1)}{(n-1) !(\kappa_{0}+\kappa_{1}+1)_{n}(\kappa_{0}+\kappa_{1}+2n)}
\Vert 1\Vert_{\mathbb{T}}^{2},
\\
&\Vert p_{4n+2,10}\Vert_{\mathbb{T}}^{2} =\frac{4\big(\kappa_{0}+\frac{1}{2}\big)_{n}\big(\kappa_{1}+\frac{3}{2}\big)
_{n}}{n!(\kappa_{0}+\kappa_{1}+1)_{n}(\kappa_{0}
+\kappa_{1}+2n+1)}\Vert 1\Vert_{\mathbb{T}}^{2},
\\
&\Vert p_{4n+2,01}\Vert_{\mathbb{T}}^{2} =\frac{4\big(\kappa_{0}+\frac{3}{2}\big)_{n}\big(\kappa_{1}+\frac{1}{2}\big)
_{n}}{n!(\kappa_{0}+\kappa_{1}+1)_{n}(\kappa_{0}
+\kappa_{1}+2n+1)}\Vert 1\Vert_{\mathbb{T}}^{2}.
\end{split}
\end{gather*}

For the odd degrees,
\begin{gather}
p_{4n+1}(z):=\bigg\{p_{4n,00}(z) +\frac{1}{4}p_{4n,11}(z)\bigg\},
\label{p4n1}
\\
p_{4n+3}(z):=z\bigg\{\bigg(n+\kappa_{0}+\frac{1}{2}\bigg) p_{4n+2,10}(z) +\bigg(n+\kappa_{1}+\frac{1}{2}\bigg)p_{4n+2,01}(z)\bigg\}.
\label{p4n3}%
\end{gather}
From the orthogonality relations $p_{4n,00}\perp p_{4n,11}$ and $p_{4n+2,10}%
\perp p_{4n+2,01}$ (different isotypes),
\begin{gather*}
\Vert p_{4n+1}\Vert_{\mathbb{T}}^{2} =\Vert p_{4n,00}\Vert_{\mathbb{T}}^{2}
+\frac{1}{16}\Vert p_{4n,11}\Vert_{\mathbb{T}}^{2}
=\frac{\big(\kappa_{0}+\frac{1}{2}\big)_{n}\big(\kappa_{1}+\frac{1}{2}\big)_{n}} {n!(\kappa_{0}+\kappa_{1}+1)_{n}}\Vert 1\Vert_{\mathbb{T}}^{2},
\\
\Vert p_{4n+3}\Vert_{\mathbb{T}}^{2} =\bigg(n+\kappa_{0}+\frac{1}{2}\bigg)^{2}
\Vert p_{4n+2,10}\Vert_{\mathbb{T}}^{2}+\bigg(n+\kappa_{1}+\frac{1}{2}\bigg)^{2}\Vert p_{4n+2,01}%
\Vert_{\mathbb{T}}^{2}
\\ \hphantom{\Vert p_{4n+3}\Vert_{\mathbb{T}}^{2}}
=4\frac{\big(\kappa_{0}+\frac{1}{2}\big)_{n+1}\big(\kappa_{1}+\frac{1}{2}\big)_{n+1}} {n!(\kappa_{0}+\kappa_{1}+1)_{n}}\Vert 1\Vert_{\mathbb{T}}^{2}.
\end{gather*}
Next we list the orthogonal basis, which involves the Laguerre polynomials.
The subscript notation may appear strange, but it makes it easy to identify
the isotype and every possibility of $(n-j,j)$ can be found by
suitably replacing $n$ (the trailing factor $g$ is understood),
\begin{gather*}
\psi_{4n+j,j}(z) =p_{4n,00}(z) L_{j}^{(\gamma_{\kappa}+4n)}( \omega z\overline{z}),
\\
\psi_{4n+2+j,j}(z) =p_{4n+2,10}(z)L_{j}^{(\gamma_{\kappa}+4n+2)}(\omega z\overline{z}),
\\
\psi_{j,4n+j}(z) =p_{4n,11}(z) L_{j}^{(\gamma_{\kappa}+4n)}( \omega z\overline{z}),
\\
\psi_{j,4n+2+j}(z) =p_{4n+2,01}(z)L_{j}^{(\gamma_{\kappa}+4n+2)}(\omega z\overline{z}).
\end{gather*}
In this list, $\sigma_{0}\psi_{2n-j,j}=\psi_{2n-j,j}$ and $\sigma_{0}%
\psi_{j,2n-j}=-\psi_{j,2n-j,}$ for $0\leq j\leq n$ ($j<n$ for the second
case). For odd degrees,
\begin{gather*}
\psi_{4n+1+j,j}(z) =p_{4n+1}(z) L_{j}^{(\gamma_{\kappa}+4n+1)}( \omega z\overline{z}), \\
\psi_{4n+3+j,j}(z) =p_{4n+3}(z) L_{j}^{(\gamma_{\kappa}+4n+3)}( \omega z\overline{z}),\\
\psi_{j,4n+1+j}(z) =\sigma_{0}\psi_{4n+1+j,j}(z),\\
\psi_{j,4n+3+j}( z) =\sigma_{0}\psi_{4n+3+j,j}(z).
\end{gather*}
By construction, $\mathcal{H}\psi_{n-j,j}=E_{n}\psi_{n-j,j}$, where the energy
eigenvalue is $E_{n}:=2\omega(n+2\kappa_{0}+2\kappa_{1}+1)$.
The squared norms of the $\psi$ follow from the formula%
\[
\Vert \psi\Vert^{2}=\frac{1}{2}\omega^{-(n+\gamma_{\kappa
}+1)}\frac{\Gamma(\gamma_{\kappa}+n+\ell+1)}{\ell
!}\Vert p\Vert_{\mathbb{T}}^{2},
\]
where $p(z) $ is homogeneous of degree $n$ and $\psi(z) =p(z) L_{\ell}^{(\gamma_{\kappa}+n)
}(\omega z\overline{z}) $. When $\omega=1$, the wavefunctions
$\psi_{n-j,j}$ are eigenfunctions of the Dunkl transform with eigenvalue
$(-\mathrm{i})^{n}$ (see \cite[Theorem~7.7.5]{DunklXu2014}).

\section{Some self-adjoint operators}

\subsection{General properties}

In this section, we are concerned with finding the action of operators which
commute with $\mathcal{H}$ on the basis functions described above. Suppose $A$
is such an operator and $A$ is self-adjoint (in~$L^{2}\big(\mathbb{R}^{2},h_{\kappa}^{2}\mathrm{d}m_{2}\big) $, then for any $(n-j,j) $ the polynomial $A\psi_{n-j,j}$ is an eigenfunction of
$\mathcal{H}$ with the same eigenvalue $2\omega(n+1+\gamma_{\kappa})$ and has an expansion $\sum_{i=0}^{n}c_{i}\psi_{n-i,i}$ (note that
$\mathrm{d}m_{2}$ denotes the $\mathbb{R}^{2}$ Lebesgue measure and equals
$r\mathrm{d}r~\mathrm{d\theta}$ for $z=r{\rm e}^{\mathrm{i}\theta}$). Suppose it is
known that the top-degree (nonzero) monomials $z^{n-i}\overline{z}^{i}$ in $A\psi
_{n-j,j}$ satisfy $m\leq i\leq N-m$, then $\psi_{n-k,k}$ for $k<m$ or
$n-m<k\leq n$ cannot appear in the expansion of $A\psi_{n-j,j}$. This is an
implicit inductive argument: if $z^{n}$ and $\overline{z}^{n}$ do not appear, then neither
$\psi_{n,0}$ nor $\psi_{0,n}$ can appear in the expansion, now consider
$z^{n-1}\overline{z}$ and $z\overline{z}^{n-1}$, $\psi_{n-1,1}$ and $\psi_{1,n-1}$ and so on. It
also follows that it suffices to consider the top-degree terms to find the
coefficients of the expansion. The top degree terms of~$\psi_{n-j,j}$ are
scalar multiples of $z^{n-j}\overline{z}^{j}\pm z^{j}\overline{z}^{n-j}$ if
$n$ is even, and of $z^{n-j}\overline{z}^{j}$ (or $z^{j}\overline{z}^{n-j}$)
if $n$ is odd and $n-j>j$ (or $n-j<j$).

\begin{Definition}
For a polynomial $p(z,\overline{z})$, let $\mathcal{C}(p,z^{m}\overline{z}^{n})$ denote the coefficient of $z^{m}\overline{z}^{n}$ in the expansion of $p$. If $p$ can be expanded in a series of
wavefunctions, then $\mathcal{C}(p,\psi_{n-j,j}) $ denotes the
coefficient of $\psi_{n-j,j}$.
\end{Definition}

Suppose, as above, that $[ A,\mathcal{H}] =0$ and $A$ is
self-adjoint, then
\begin{equation}
\mathcal{C}(A\psi_{n-j,j},\psi_{n-k,k}) \Vert \psi
_{n-k,k}\Vert^{2}=\overline{\mathcal{C}}(A\psi_{n-k,k},\psi_{n-j,j}) \Vert \psi_{n-j,j}\Vert^{2};
\label{CAij}%
\end{equation}
generally the coefficients we use are real and the complex conjugate on
$\mathcal{C}$ can be omitted. In~particular, $\mathcal{C}(A\psi_{n-j,j},\psi_{n-k,k}) =0$ implies $\mathcal{C}(A\psi_{n-k,k},\psi_{n-j,j}) =0$.

\subsection{\label{AngMo}Angular momentum}

Here $\mathcal{J}=zT-\overline{z}\overline{T}$, from Proposition~\ref{Jprops}
we have $\mathcal{J}^{\ast}=-\mathcal{J}$ and $[ \mathcal{J},\mathcal{H}] =0$. Also $g^{-1}\mathcal{J}g=\mathcal{J}$ (because
$\mathcal{J}(z\overline{z})^{k}=0$). We determine the effect of
$\mathcal{J}$ on $\psi_{n-j,j}$ by considering the dominant top-degree
monomials. This suffices because $[ \mathcal{J},\Delta_{\kappa}]
=0$ and the image of a $\kappa$-harmonic polynomial under $\mathcal{J}$ is
$\kappa$-harmonic, and there are only two (independent) $\kappa$-harmonic
polynomials of each degree ($\geq1$)
\begin{align*}
\mathcal{J}z^{n} &=nz^{n}+\sum_{j=0}^{3}\kappa_{j\operatorname{mod}2}%
\frac{z^{n}-(\mathrm{i}^{j}\overline{z})^{n}}{z-\mathrm{i}%
^{j}\overline{z}}(z+\mathrm{i}^{j}\overline{z})
\\
 &=nz^{n}+2(\kappa_{0}+\kappa_{1}) z^{n}+\big\{ \kappa
_{0}(1+(-1)^{n}) +\kappa_{1}\big(\mathrm{i}^{n}+(-\mathrm{i})^{n}\big) \big\} \overline
{z}^{n}+\cdots
\\
 &=(n+2\kappa_{0}+2\kappa_{1}) z^{n}+\big(1+(-1)^{n}\big) \big( \kappa_{0}+\mathrm{i}^{n}\kappa_{1}\big)
\overline{z}^{n}+\cdots,
\end{align*}
omitting terms like $z^{n-j}\overline{z}^{j}$ with $1\leq j<n$. Also
\[
\mathcal{J}\overline{z}^{n}=-(n+2\kappa_{0}+2\kappa_{1})
\overline{z}^{n}-\big(1+(-1)^{n}\big) \big(\kappa_{0}+\mathrm{i}^{n}\kappa_{1}\big) z^{n}+\cdots
\]
because $\sigma_{0}\mathcal{J}\sigma_{0}=-\mathcal{J}$. If $n$ is odd, then
\[
\mathcal{J}z^{n}=(n+2\kappa_{0}+2\kappa_{1}) z^{n}+\cdots\qquad \text{and}\qquad
\mathcal{J}^{2}z^{n}=(n+2\kappa_{0}+2\kappa_{1})^{2}z^{n}+\cdots.
\]
Thus
\begin{align*}
&\mathcal{J}\psi_{n+j,j} =\mathcal{J}\big(p_{n}(z)
L_{j}^{(\gamma_{\kappa}+n)}(\omega z\overline{z})
\big) =(\mathcal{J}p_{n}(z) ) L_{j}^{(\gamma_{\kappa}+n)}( \omega z\overline{z})
 =(n+2\kappa_{0}+2\kappa_{1}) \psi_{n+j,j},
\\
&\mathcal{J}\psi_{j,n+j} =-(n+2\kappa_{0}+2\kappa_{1})\psi_{j,n+j}.
\end{align*}
If $n=0\operatorname{mod}4$, then
\begin{align*}
&\mathcal{J}z^{n} =(n+2\kappa_{0}+2\kappa_{1}) z^{n}+2(\kappa_{0}+\kappa_{1}) \overline{z}^{n}+\cdots,
\\
&\mathcal{J}z^{n} =-(n+2\kappa_{0}+2\kappa_{1}) \overline
{z}^{n}-2(\kappa_{0}+\kappa_{1}) \overline{z}^{n}+\cdots
\end{align*}
and thus%
\begin{align*}
& \mathcal{J}^{2}z^{n} =(n+2\kappa_{0}+2\kappa_{1})
\mathcal{J}z^{n}+2(\kappa_{0}+\kappa_{1}) \mathcal{J}
\overline{z}^{n} =n(n+4\kappa_{0}+4\kappa_{1}) z^{n}+\cdots,
\\
& \mathcal{J}^{2}\overline{z}^{n} =n(n+4\kappa_{0}+4\kappa_{1}) \overline{z}^{n}+\cdots.
\end{align*}

If $n=2\operatorname{mod}4$, then
\begin{align*}
& \mathcal{J}z^{n} =(n+2\kappa_{0}+2\kappa_{1}) z^{n}+2(\kappa_{0}-\kappa_{1}) \overline{z}^{n}+\cdots,
\\
 & \mathcal{J}z^{n}=-(n+2\kappa_{0}+2\kappa_{1}) \overline
{z}^{n}-2(\kappa_{0}-\kappa_{1}) \overline{z}^{n}+\cdots,
\\
&\mathcal{J}^{2}z^{n} =(n+4\kappa_{0}) (n+4\kappa_{1}) z^{n}+\cdots,
\\
&\mathcal{J}^{2}\overline{z}^{n} =(n+4\kappa_{0}) (n+4\kappa_{1}) \overline{z}^{n}+\cdots.
\end{align*}
Thus
\begin{align}
&\mathcal{J}^{2}\psi_{4n+j,j} =16n(n+\kappa_{0}+\kappa_{1})\psi_{4n+j,j},\label{J4n}
\\
&\mathcal{J}^{2}\psi_{j,4n+j} =16n(n+\kappa_{0}+\kappa_{1})\psi_{j,4n+j},\label{J4nR}
\\
&\mathcal{J}^{2}\psi_{4n+2+j,j} =4(2n+2\kappa_{0}+1) (2n+2\kappa_{1}+1) \psi_{4n+2+j,j},\label{J4n2}
\\
&\mathcal{J}^{2}\psi_{j,4n+2+j} =4(2n+2\kappa_{0}+1) (2n+2\kappa_{1}+1) \psi_{j,4n+2+j}. \label{J4n2R}%
\end{align}

\subsection{Raising and lowering operators}

Formula \eqref{gHg} specializes to%
\[
\widehat{H}_{j}=\mathrm{i}^{j}T^{2}-2T\overline{T}+\mathrm{i}^{-j}\overline
{T}^{2}+\frac{\omega}{2}\big\{ \big(z-\mathrm{i}^{j}\overline{z}\big)
\big(T-\mathrm{i}^{-j}\overline{T}\big) +\big(T-\mathrm{i}
^{-j}\overline{T}\big) \big(z-\mathrm{i}^{j}\overline{z}\big)\big\}.
\]

\begin{Proposition}
$H_{0}+H_{2}=\mathcal{H}=H_{1}+H_{3}$.
\end{Proposition}

\begin{proof}
We use the polynomial parts $\widehat{H}_{j}$. First%
\begin{align*}
&\widehat{H}_{2j} =-2T\overline{T}+\frac{\omega}{2}\big\{ zT+Tz+\overline
{z}\overline{T}+\overline{T}\overline{z}\big\}
 +(-1)^{j}\bigg(T^{2}+\overline{T}^{2}-\frac{\omega}
{2}\big\{ z\overline{T}+\overline{z}T+\overline{T}z+T\overline{z}\big\}\bigg),
\end{align*}
thus $\widehat{H}_{0}+\widehat{H}_{2}=-4T\overline{T}+\omega\big\{
zT+Tz+\overline{z}\overline{T}+\overline{T}\overline{z}\big\}
=g^{-1}\mathcal{H}g$. Similarly, $\widehat{H}_{1}+ \widehat{H}_{3}%
=g^{-1}\mathcal{H}g$ (using $\mathrm{i}^{2j+1}=(-1)^{j}\mathrm{i}$).
\end{proof}

\begin{Corollary}
$[H_{0},H_{2}] =0$ and $[H_{1},H_{3}] =0$.
\end{Corollary}

\begin{proof}
$H_{0}H_{2}=H_{0}(\mathcal{H}-H_{0})$ and $[H_{0},\mathcal{H}] =0$ by Proposition~\ref{[hH'0}.
\end{proof}

The formula $w^{-1}H_{a}w=H_{aw}$ (see Proposition~\ref{wHaw}) implies
$\rho_{3}\widehat{H}_{0}\rho_{1}=\widehat{H}_{2}$. To find the effect of~$\rho_{1}$ or $\rho_{3}$ consider the leading term in $\psi_{m+j,j}=p_{m}(z) L_{j}^{(\gamma_{\kappa}+m)}(\omega z\overline{z})$ (total degree $n=m+2j$), a~scalar multiple of
$z^{m+\ell}\overline{z}^{\ell}$; examination of each of the formulas for
$p_{m}$ shows that each monomial $z^{a}\overline{z}^{b}$ satisfies
$a-b=m\operatorname{mod}4$, this also applies to $\psi_{j,m+j}$ except for
the odd case where the leading term is $z^{j}\overline{z}^{j+m}$ and
$a-b=-m\operatorname{mod}4$. Also $\rho_{1}(z^{a}\overline{z}^{b}) =(\mathrm{i}z)^{a}(-\mathrm{i}\overline
{z})^{b}=\mathrm{i}^{a-b}z^{a}\overline{z}^{b}=\mathrm{i}^{m}z^{a}\overline{z}^{b}$.
Thus by replacing $m$ by $n-2j$, we obtain $\rho
_{1}\psi_{n-j,j}=\mathrm{i}^{n-2j}\psi_{n-j,j}$; this applies to all $n$ (if
$n$ is even then $\mathrm{i}^{n-2j}=\mathrm{i}^{2j-n}$). Replace $\mathrm{i}$
by $-\mathrm{i}$ to find $\rho_{3}\psi_{n-j,j}$.

\begin{Proposition}
\label{H0H2}
Suppose $\widehat{H}_{0}\psi_{n-j,j}=\sum_{j=0}^{n}c_{j,i}\psi_{n-i,i}$, then $(1)$ $c_{j,j}=\frac{1}{2}E_{n}$, $(2)$ $i=j\operatorname{mod}2$ and $i\neq j$ implies $c_{i,j}=0$, $(3)$ $\widehat{H}_{2}\psi_{n-j,j}=\frac{1}{2}E_{n}\psi_{n-j,j}-\sum\{c_{j,i}\psi_{n-i,i}
\colon j-i=1\operatorname{mod}2\} $.
\end{Proposition}

\begin{proof}
By hypothesis,
\begin{align*}
H_{2}\psi_{n-2j,j} & =\rho_{3}H_{0}\rho_{1}\psi_{n-j,j}
=\rho_{3}\mathrm{i}^{n-2j}\sum_{i=0}^{n}c_{j,i}\psi_{n-i,i}
=\mathrm{i}^{n-2j}\sum_{i=0}^{n}c_{j,i}\mathrm{i}^{-n+2i}\psi_{n-i,i}
\\
&=\sum_{i=0}^{n}(-1)^{i-j}c_{j,i}\psi_{n-i,i}.
\end{align*}
From $H_{0}+H_{2}=\mathcal{H}$, it follows that $\sum_{i=0}^{n}c_{j,i}\big(1+(-1)^{j-i}\big) \psi_{n-i,i}=E_{n}\psi_{n-j,j}$. Thus
$2c_{j,j}=E_{n}$ and $j-i=0\operatorname{mod}2$ and $j\neq i$ implies
$c_{j,i}=0$.
\end{proof}

Since $\{ \sigma_{0},\sigma_{2}\} $ and $\{ \sigma
_{1},\sigma_{3}\} $ are conjugacy classes, the operators $H_{0}%
^{m}+H_{2}^{m}$ and $H_{1}^{m}+H_{3}^{m}$ commute with $\mathcal{H}$ and with
the action of the group for $m=1,2,3,\ldots$ (Proposition~\ref{wHaw}). There
is a fairly simple formula for $\sum_{i=0}^{3}H_{i}^{2}$. Note $I$ denotes the
identity in the group.

\begin{Definition}
Set $R:=(I+\kappa_{0}(\sigma_{0}+\sigma_{2}) +\kappa
_{1}(\sigma_{1}+\sigma_{3}) )^{2}-2\big(\kappa_{0}^{2}+\kappa_{1}^{2}\big) (1-\rho_{2})$. This is an
element of the center of the group algebra, that is, $[R,\sigma
_{j}] =0$ for $0\leq j\leq3$ and $[R,\rho_{j}] =0$ for $1\leq j\leq3$.
\end{Definition}

Equivalently, $R=I+4\big(\kappa_{0}^{2}+\kappa_{1}^{2}\big) \rho
_{2}+2\kappa_{0} (\sigma_{0}+\sigma_{2} ) +2\kappa_{1} (\sigma_{1}+\sigma_{3} ) +4\kappa_{0}\kappa_{1} (\rho_{1}+\rho
_{3} )$.

\begin{Theorem}\label{H0123sq}$H_{0}^{2}+H_{2}^{2}+H_{1}^{2}+H_{3}^{2}=\frac{3}{2}%
\mathcal{H}^{2}-2\omega^{2}\mathcal{J}^{2}-2\omega^{2}R$.
\end{Theorem}

\begin{proof}
The details are presented in Appendix~\ref{Symb}. The idea is to use direct
(computer-assisted) calculation.
\end{proof}

The following defines the operator which is the main concern in the sequel; it
will be shown to commute with $\mathcal{H}$ and the group action but not with
angular momentum. The latter claim is proven by demonstrating that
eigenfunctions of $\mathcal{J}^{2}$ are not preserved.

\begin{Definition}
Set $\mathcal{K}:=H_{0}^{2}+H_{2}^{2}-H_{1}^{2}-H_{3}^{2}$, a fourth-order operator.
\end{Definition}

From Propositions~\ref{[hH'0} and~\ref{wHaw}, it follows that $[\mathcal{K},\mathcal{H}] =0$ and each of $H_{0}^{2}+H_{2}^{2}$, $H_{1}^{2}+H_{3}^{2}$ commutes with the group action.

\begin{Proposition}
$\mathcal{K}=2(H_{1}H_{3}-H_{0}H_{2}) =-\frac{1}{2}
\mathcal{H}^{2}+2\omega^{2}\mathcal{J}^{2}+2\mathcal{\omega}^{2}R+4\big(H_{0}-\frac{1}{2}\mathcal{H}\big)^{2}$.
\end{Proposition}

\begin{proof}
$\mathcal{K}=(H_{0}+H_{2})^{2}-2H_{0}H_{2}-(H_{1}%
+H_{3})^{2}+2H_{1}H_{3}$. Also $\mathcal{K}+\sum_{i=0}^{3}H_{i}%
^{2}=2\big(H_{0}^{2}+H_{2}^{2}\big) =2\mathcal{H}^{2}-4H_{0}H_{2}$. From
Proposition~\ref{H0H2}, $H_{2}-\frac{1}{2}\mathcal{H=-}(H_{0}-\frac
{1}{2}\mathcal{H}) $ and $H_{2}H_{0}=\frac{1}{4}\mathcal{H}^{2}-\big(H_{0}-\frac{1}{2}\mathcal{H}\big)^{2}$. Thus%
\begin{align*}
\mathcal{K} & =2\mathcal{H}^{2}-4H_{0}H_{2}-\bigg(\frac{3}{2}\mathcal{H}^{2}-2\omega^{2}\mathcal{J}^{2}-2\omega^{2}R\bigg)
\\
& =\frac{1}{2}\mathcal{H}^{2}+2\omega^{2}\mathcal{J}^{2}+2\omega
^{2}R-4\bigg\{ \frac{1}{4}\mathcal{H}^{2}-\bigg(H_{0}-\frac{1}{2}\mathcal{H}\bigg)^{2}\bigg\};
\end{align*}
this completes the proof.
\end{proof}

\section[The expansion coefficients of H\_0]{The expansion coefficients of $\boldsymbol{H_{0}}$}\label{ExpH0}

\subsection{General formulas}

This section calculates the coefficients in $H_{0}\psi_{n-j,j}$. Start with%
\begin{align*}
\begin{split}
\widehat{H}_{0} & =\big(T-\overline{T}\big)^{2}+\frac{\omega}%
{2}\big\{ (z-\overline{z}) \big(T-\overline{T}\big)
+\big(T-\overline{T}\big) (z-\overline{z}) \big\}
\\
& =\big(T-\overline{T}\big)^{2}+\omega\big\{ (z-\overline
{z}) \big(T-\overline{T}\big) +1+2\kappa_{0}\sigma_{0}+\kappa
_{1}(\sigma_{1}+\sigma_{3}) \big\}.
\end{split}
\end{align*}
Let%
\begin{gather*}
A :=(z-\overline{z}) (\partial_{z}-\partial_{\overline{z}}) +1+2\kappa_{0},
\\
B :=\kappa_{1}\bigg\{ (\sigma_{1}+\sigma_{3}) +(z-\overline{z}) \bigg\{ \frac{(1+\mathrm{i})
}{z-\mathrm{i}\overline{z}}(1-\sigma_{1}) +\frac{(1-\mathrm{i})}{z+\mathrm{i}\overline{z}}(1-\sigma_{3})\bigg\} \bigg\}
\end{gather*}
so that $\widehat{H}_{0}=(T-\overline{T})^{2}+\omega(A+B)$. The part of $(z-\overline{z}) (T-\overline{T})$ corresponding to $\sigma_{0}$ is $2\kappa_{0}({z-\overline{z}})\times \frac{1-\sigma_{0}}{(z-\overline{z})}=2\kappa_{0}(1-\sigma_{0})$. To determine the coefficients in
$\widehat{H}_{0}\psi_{n-j,j}=\sum_{i=0}^{n}c_{j,i}\psi_{n-i,i}$, it
suffices to consider the monomials of degree $n$ in $\widehat{H}_{0}\psi_{n-j,j}$, that is, analyze $(A+B) \psi_{j,n-j}$. For a polynomial $p=\sum_{i=k}^{\ell}c_{i}z^{n-i}\overline{z}^{i}$ with $c_{k}\neq0\neq c_{\ell}$ let $D(p) =c_{k}z^{n-k}\overline{z}^{k}+c_{\ell}z^{n-\ell}\overline{z}^{\ell}$ ($D$ for dominant terms). Then
(for $2j\leq n$)
\begin{gather*}
D\big(Az^{n-j}\overline{z}^{j}\big) =-jz^{n-j+1}\overline{z}%
^{j-1}-(n-j) z^{n-j-1}\overline{z}^{j+1},\\
D\big(Az^{j}\overline{z}^{n-j}\big) =-jz^{j-1}\overline{z}%
^{n-j+1}-(n-j) z^{j+1}\overline{z}^{n-j-1}.
\end{gather*}
Let $2j\leq n$ and $m:=n-2j$, then%
\begin{gather*}
Bz^{n-j}\overline{z}^{j} = \kappa_{1}(z\overline{z})
^{j}\mathrm{i}^{m}(1+(-1)^{m}) \overline{z}^{m}\\
\hphantom{Bz^{n-j}\overline{z}^{j}=}{}+\kappa_{1}(z\overline{z})^{j}(z-\overline{z})
\sum_{k=0}^{m-1}z^{m-1-k}\overline{z}^{k}\big\{ \mathrm{i}^{k}(1+\mathrm{i}) +( -\mathrm{i})^{k}(1-\mathrm{i}%
) \big\} \\
\hphantom{Bz^{n-j}\overline{z}^{j}}{}= \kappa_{1}(z\overline{z})^{j}\bigg\{ 2z^{m}+2\varepsilon
_{m-1}\overline{z}^{m}+2\sum_{k=1}^{m-1}z^{m-k}\overline{z}^{k}(\varepsilon_{k}-\varepsilon_{k-1}) \bigg\},
\end{gather*}
where $\varepsilon_{k}=(-1)^{\lfloor (k+1)/2\rfloor}$ and $\lfloor r\rfloor $ is the largest integer
$\leq r$. Thus
\[
D\big(Bz^{n-j}\overline{z}^{j}\big) =2\kappa_{1}\big(z^{n-j}\overline{z}^{j}+ \varepsilon_{n-2j-1}z^{j}\overline{z}^{n-j}\big);
\]
the same formula holds with $(z,\overline{z}) $ replaced by
$(\overline{z},z) $ (because $[ B,\sigma_{0}] =0$).
The special case $Bz^{j}\overline{z}^{j}=2\kappa_{1}z^{j}\overline{z}^{j}$.

Let $1\leq j\leq n$, then by construction $D(\psi_{2n-j,j})
=\mathcal{C}\big(\psi_{2n-j,j},z^{2n-j}\overline{z}^{j}\big)
 \big(z^{2n-j}\overline{z}^{j}+z^{j}\overline{z}^{2n-j}\big)
$, similarly for $\psi_{j,2n-j}$; and from the above formulas, it follows that
\begin{align}
&D((A+B) \psi_{2n-j,j}) =-j\mathcal{C}(\psi_{2n-j,j},z^{2n-j}\overline{z}^{j}) (z^{2n-j+1}\overline
{z}^{j-1}+z^{j-1}\overline{z}^{2n-j+1}), \label{DAB1}\\
&D((A+B) \psi_{j,2n-j}) =-j\mathcal{C}(\psi_{j,2n-j},z^{2n-j}\overline{z}^{j}) (z^{2n-j+1}\overline
{z}^{j-1}-z^{j-1}\overline{z}^{2n-j+1}), \label{DAB2}%
\end{align}
the other terms are dominated by these. This implies that
$\mathcal{C}(H_{0}\psi_{2n-j,j},\psi_{2n-i,i}) =0$ and
$\mathcal{C}(H_{0}\psi_{j,2n-j},\psi_{i,2n-i,}) =0$ for $i<j-1$.
Thus the nonzero coefficients occur only for $\vert j-i\vert \leq1$.

For the odd case, suppose $j\leq n$, then
\begin{gather*}
D(\psi_{2n+1-j,j}) =\mathcal{C}\big(\psi_{2n+1-j,j},z^{2n+1-j}\overline{z}^{j}\big) z^{2n+1-j}\overline{z}^{j}
 +\mathcal{C}\big(\psi_{2n+1-j,j},z^{j+1}\overline{z}^{2n-j}\big)
z^{j+1}\overline{z}^{2n-j}.
\end{gather*}
This implies%
\begin{gather}
D((A+B) \psi_{2n+1-j,j}) = \bigl(-jz^{2n+2-j}\overline{z}^{j-1}+2\varepsilon_{2n+2j}\kappa_{1}z^{j}\overline
{z}^{2n+1-j}\bigr)\mathcal{C}\big(\psi_{2n+1-j,j},z^{2n+1-j}\overline{z}^{j}\big)\nonumber
\\
\hphantom{D((A+B) \psi_{2n+1-j,j}) =}{}
 -(j+1) \mathcal{C}\big(\psi_{2n+1-j,j},z^{j+1}\overline
{z}^{2n-j}\big) z^{j}\overline{z}^{2n+1-j}.\label{oddD}
\end{gather}

\subsection{Even degree}

To find $\mathcal{C}\big(\widehat{H}_{0}\psi_{2n-j,j},\psi_{2n-j+1,j-1}\big) $ and $\mathcal{C}%
\big(\widehat{H}_{0}\psi_{j,2n-j},\psi_{j-1,2n-j+1}\big)$, we used formulas \eqref{DAB1} and \eqref{DAB2}.
Then for $\mathcal{C}\big(\widehat{H}_{0}\psi_{2n-j,j},\psi_{2n-j-1,j+1}\big)$ and $\mathcal{C}\big(\widehat{H}_{0}\psi_{j,2n-j},\psi_{j+1,2n-j-1}\big)$, we used~\eqref{CAij}. The leading coefficients of $\psi_{n-j,j}$ are derived from%
\begin{align*}
 &\mathcal{C}\big(R_{n}^{(\alpha,\beta)},z^{4n}\big)
=\mathcal{C}\big(R_{n}^{(\alpha,\beta)},\overline{z}^{4n}\big) =\frac{1}{2^{2n}n!}(n+\alpha+\beta+1)_{n},
\\
&\mathcal{C}\big(L_{j}^{(n+\gamma_{\kappa})}(\omega
z\overline{z}),z^{j}\overline{z}^{j}\big) =(-1)^{j}\frac{\omega^{j}}{j!}.
\end{align*}
Then%
\begin{align*}
 &\mathcal{C}\big(p_{4n,00},z^{4n}\big) =\mathcal{C}\big(p_{4n,00},\overline{z}^{4n}\big) =\frac{1}{2^{2n}n!}(n+\kappa
_{0}+\kappa_{1})_{n},
\\
&\mathcal{C}\big(p_{4n,11},z^{4n}\big) =-\mathcal{C}\big(p_{4n,11},\overline{z}^{4n}\big) =\frac{1}{2^{2n-2}(n-1)
!}(n+\kappa_{0}+\kappa_{1}+1)_{n-1}%
\end{align*}
and%
\begin{align*}
&\mathcal{C}\big(p_{4n,+2,10},z^{4n+2}\big) =\mathcal{C}\big(p_{4n+2,10},\overline{z}^{4n+2}\big) =\frac{1}{2^{2n}n!}(n+\kappa_{0}+\kappa_{1}+1)_{n},
\\
&\mathcal{C}\big(p_{4n+2,01},z^{4n+2}\big) =-\mathcal{C}\big(p_{4n+2,01},\overline{z}^{4n+2}\big) =\frac{1}{2^{2n}n!}(n+\kappa_{0}+\kappa_{1}+1)_{n}.
\end{align*}
First consider the even degree polynomials satisfying $\sigma_{0}p=p$%
\begin{gather}
\left(H_{0}-\frac{1}{2}E_{4n+2j}\right) \psi_{4n+j,j}=\omega^{2}%
\frac{n+\kappa_{0}+\kappa_{1}}{2n+\kappa_{0}+\kappa_{1}}\psi_{4n+j+1,j-1}%
\nonumber
\\ \qquad
{}+\frac{(j+1) (2n+2\kappa_{0}-1) (4n+2\kappa_{0}+2\kappa_{1}+j)}{2(2n+\kappa_{0}+\kappa
_{1})}\psi_{4n+j-1,j+1},\label{cof4n}
\\
\left(H_{0}-\frac{1}{2}E_{4n+2j+2}\right) \psi_{4n+2+j,j}=4\omega^{2}%
\frac{n+1}{2n+\kappa_{0}+\kappa_{1}+1}\psi_{4n+j+3,j-1}\nonumber
\\ \qquad
{}+2\frac{(j+1) (2n+2\kappa_{1}+1) (4n+2\kappa_{0}+2\kappa_{1}+j+2)}{2n+\kappa_{0}+\kappa_{1}+1}%
\psi_{4n+1+j,j+1},\label{cof4n2}
\end{gather}
then the even degree polynomials satisfying $\sigma_{0}p=-p$,%
\begin{gather}
\left(H_{0}-\frac{1}{2}E_{4n+2j}\right) \psi_{j,4n+j}=4\omega^{2}\frac
{n}{2n+\kappa_{0}+\kappa_{1}}\psi_{j-1,4n+j+1}\nonumber
\\ \qquad
{}+\frac{(j+1) (2n+2\kappa_{1}-1) (4n+2\kappa_{0}+2\kappa_{1}+j)}{( 2n+\kappa_{0}+\kappa
_{1})}\psi_{j+1,4n+j-1},\label{cof4nR}
\\
\left(H_{0}-\frac{1}{2}E_{4n+2j+2}\right) \psi_{j,4n+2+j}=\omega^{2}%
\frac{n+\kappa_{0}+\kappa_{1}+1}{2n+\kappa_{0}+\kappa_{1}+1}\psi
_{j-1,4n+j+3}\nonumber
\\ \qquad
{}+\frac{(j+1) (2n+2\kappa_{0}+1) (4n+2\kappa_{0}+2\kappa_{1}+j+2)}{2(2n+\kappa_{0}+\kappa
_{1}+1)}\psi_{j+1,4n+1+j}.\label{cof4n2R}
\end{gather}
Thus the matrix of $H_{0}$ in the bases $\{\psi_{2n-j,j}\colon 0\leq j\leq
n\}$ and $\{\psi_{j,2n-j}\colon 0\leq j<n\} $ is tridiagonal.

\subsection{Odd degree}

Formula \eqref{oddD} is used to find $\mathcal{C}\big(\widehat{H}_{0}
\psi_{2n+1-j,j},\psi_{2n+2-j,j-1}\big)$ and $\mathcal{C}\big(\widehat{H}_{0}\psi_{2n+1-j,j},\psi_{j,2n+1-j}\big) $. We will show that
the nonzero coefficients in $H_{0}\psi_{2n+1-j,j}=\sum_{i=0}^{2n+1}c_{j,i}
\psi_{2n+1-i,i}$ occur at $i=j-1,\allowbreak j,j+1,2n+1-j$. Suppose $m$ is odd, then
$\psi_{m-j,j}=\sum_{i=j}^{m-j-1}a_{i}z^{m-i}\overline{z}^{i}$ and
$a_{j},a_{m-j-1}$ are involved in finding $\mathcal{C}\big(H_{0}
\psi_{m-j,j},z^{m-j+1}\overline{z}^{j-1}\big)$ and
$\mathcal{C}\big(H_{0}\psi_{m-j,j},z^{j}\overline{z}^{m-j}\big)$. Then
\begin{gather*}
\mathcal{C}(H_{0}\psi_{m+j,j},\psi_{m+1+j,j-1}) =\frac
{\mathcal{C}\big(H_{0}\psi_{m+j,j},z^{m+1+j}\overline{z}^{j-1}\big)
}{\mathcal{C}\big(\psi_{m+1+j,j-1},z^{m+1+j}\overline{z}^{j-1}\big)}
\\ \hphantom{\mathcal{C}(H_{0}\psi_{m+j,j},\psi_{m+1+j,j-1})}
{}=-\frac{j\omega\mathcal{C}\big(\psi_{m+j,j},z^{m+j}\overline{z}^{j}\big)
}{\mathcal{C}\big(\psi_{m+1+j,j-1},z^{m+1+j}\overline{z}^{j-1}\big)},
\\
\mathcal{C}\big(H_{0}\psi_{m+j,j},z^{j}\overline{z}^{m+j}\big)
=\mathcal{C}\big(H_{0}\psi_{m+j,j},\psi_{m+1+j,j-1}\big) \mathcal{C}%
\big(\psi_{m+1+j,j-1},z^{j}\overline{z}^{m+j}\big)
\\ \hphantom{\mathcal{C}\big(H_{0}\psi_{m+j,j},z^{j}\overline{z}^{m+j}\big)=}
{}+\mathcal{C}(H_{0}\psi_{m+j,j},\psi_{j,m+j}) \mathcal{C}\big(\psi_{j,m+j},z^{j}\overline{z}^{m+j}\big)
\\ \hphantom{\mathcal{C}\big(H_{0}\psi_{m+j,j},z^{j}\overline{z}^{m+j}\big)}
{}=-\omega(j+1) \mathcal{C}\big(\psi_{m+j,j},z^{j+1}\overline{z}^{m-1+j}\big)
\\ \hphantom{\mathcal{C}\big(H_{0}\psi_{m+j,j},z^{j}\overline{z}^{m+j}\big)=}
{}+2(-1)^{(m+1)
/2}\omega\kappa_{1}\mathcal{C}\big(\psi_{m+j,j},z^{m+j}\overline{z}%
^{j}\big).
\end{gather*}
The lower three lines are used to solve for $\mathcal{C} (H_{0}%
\psi_{m+j,j},\psi_{j,m+j} ) $. There are two cases: $m-2j=4n+1,4n+3$.
First%
\begin{gather*}
\mathcal{C}\big(p_{4n+1},z^{4n+1}\big) =\mathcal{C}\big(p_{4n,00},z^{4n}\big)
+\frac{1}{4}\mathcal{C}\big(p_{4n,11},z^{4n}\big) =\frac{1}{2^{2n}n!}(n+\kappa_{0}+\kappa_{1}+1)_{n},
\\
\mathcal{C}\big(p_{4n+1},z\overline{z}^{4n}\big) =\mathcal{C}\big(p_{4n,00},z^{4n}\big) -\frac{1}{4}\mathcal{C}\big(p_{4n,11},z^{4n}\big)
=\frac{\kappa_{0}+\kappa_{1}}{2^{2n}n!}(n+\kappa_{0}+\kappa_{1}+1)_{n-1}
\end{gather*}
and second
\begin{gather*}
\mathcal{C}\big(p_{4n+3},z^{4n+3}\big) =\bigg(n+\kappa_{0}+\frac{1}{2}\bigg) \mathcal{C}\big(p_{4n+2,10},z^{4n+2}\big) +\bigg(n+\kappa_{1}+\frac{1}{2}\bigg)
\mathcal{C}\big(p_{4n+2,01},z^{4n+2}\big)
\\ \hphantom{\mathcal{C}\big(p_{4n+3},z^{4n+3}\big)}
{} =\frac{1}{2^{2n}n!}(n+\kappa_{0}+\kappa_{1}+1)_{n+1},
\\
\mathcal{C}\big(p_{4n+3},z\overline{z}^{4n+2}\big) =\bigg(n+\kappa_{0}+\frac{1}{2}\bigg) \mathcal{C}\big(p_{4n+2,10},z^{4n+2}\big) -\bigg(n+\kappa_{1}+\frac{1}{2}\bigg)
\mathcal{C}\big(p_{4n+2,01},z^{4n+2}\big)
\\ \hphantom{\mathcal{C}\big(p_{4n+3},z\overline{z}^{4n+2}\big) }
 {}=\frac{\kappa_{0}-\kappa_{1}}{2^{2n}n!}(n+\kappa_{0}+\kappa_{1}+1)_{n}.
\end{gather*}
Then%
\begin{gather}
\bigg(H_{0}-\frac{1}{2}E_{4n+1+2j}\bigg) \psi_{4n+1+j,j}\nonumber
\\ \qquad
{}=\frac{\omega^{2}}{2n+\kappa_{0}+\kappa_{1}+1}\psi_{4n+2+j,j-1}
+\frac{(j+1) (4n+2\kappa_{0}+2\kappa_{1}+j+1)
}{2n+\kappa_{0}+\kappa_{1}}\psi_{4n+j,j+1}\nonumber
\\ \qquad
\phantom{=}{}+\omega\bigg\{ \frac{j(2n+2\kappa_{0}+1)}{2n+\kappa_{0}%
+\kappa_{1}+1}+\frac{2n(j+1)}{2n+\kappa_{0}+\kappa_{1}}-(2j+2\kappa_{1}+1)\bigg \} \psi_{j,4n+1+j},\label{cof4n1}
\\[1mm]
\bigg(H_{0}-\frac{1}{2}E_{4n+3+2j}\bigg) \psi_{4n+3+j,j}\nonumber
\\ \qquad
{}=\frac{4\omega^{2}(n+1) (n+\kappa_{0}+\kappa_{1}+1)}%
{2n+\kappa_{0}+\kappa_{1}+2}\psi_{4n+4+j,j-1}\nonumber
\\ \qquad
\phantom{=}{}+\frac{(j+1) (4n+2\kappa_{0}+2\kappa_{1}+j+3)
(2n+2\kappa_{0}+1) (2n+2\kappa_{1}+1)}%
{2n+\kappa_{0}+\kappa_{1}+1}\psi_{4n+2+j,j+1}\nonumber
\\ \qquad
\phantom{=}{}+\omega\bigg\{ \frac{(j+1) (2n+2\kappa_{0}+1)
}{2n+\kappa_{0}+\kappa_{1}+1}-\frac{2j(n+1)}{2n+\kappa
_{0}+\kappa_{1}+2}+(2j+2\kappa_{1}+1) \bigg\} \psi
_{j,4n+3+j}.\label{cof4n3}
\end{gather}

\section[The expansion coefficients of K]{The expansion coefficients of $\boldsymbol{\mathcal{K}}$}\label{ExpK}

\subsection{Even degree}

The matrices of $H_{0}-\frac{1}{2}\mathcal{H}$ with respect to the bases
$\{ \psi_{2n-j,j}\colon 0\leq j\leq n\}$ ($\sigma_{0}p=p$) and
$\{ \psi_{j,2n-j}\colon 0\leq j<n\}$ ($\sigma_{0}p=-p$) are
tridiagonal with zeroes on the main diagonal. If $M$ is such a matrix, then the
only nonzero elements in $M^{2}$ are $\big(M^{2}\big)_{i,i-2}%
=M_{i,i-1}M_{i-1,i-2}$, $\big(M^{2}\big)_{i,i+2}=M_{i,i+1}M_{i+1,i+2}$ and
$\big(M^{2}\big)_{i,i}=M_{i,i-1}M_{i-1,i}+M_{i,i+1}M_{i+1,i}$. By use
of the expansion coefficients of $H_{0}-\frac{1}{2}\mathcal{H}$, we find the
matrix for $\mathcal{K}=-\frac{1}{2}\mathcal{H}^{2}+2\omega^{2}\mathcal{J}%
^{2}+2\mathcal{\omega}^{2}R+4\big(H_{0}-\frac{1}{2}\mathcal{H}\big)
^{2}$ (the first three terms act as scalars).

To compute $\mathcal{K}$ $\psi_{4n+j,j}$, we use the coefficients of the
expansions of$H_{0}\psi_{4n+j+k,j-k}$ with $k=-1,0,1$ from
\eqref{cof4n}, \eqref{cof4n2}. Note $4n+j-1=(4(n-1)
+2) +(j+1) $ and $4n+j+1=(4n+2) +(j-1) $; this indicates which types are involved. The scalars derive
from $E_{4n+2j}$, $R\psi_{4n+j,j}=(1+2\kappa_{0}+2\kappa_{1})
^{2}\psi_{4n+j,j}$ and $\mathcal{J}^{2}\psi_{4n+j,j}=16n(n+\kappa
_{0}+\kappa_{1}) \psi_{4n+j,j}$ (from \eqref{J4n})
\begin{gather*}
\mathcal{K}\psi_{4n+j,j} =A_{-1}^{0}(n) \psi_{4n+j+2,j-2}%
+A_{0}^{0}(n,j) \psi_{4n+j,j}+A_{1}^{0}(n,j)\psi_{4n+j-2,j+2},
\\
A_{-1}^{0}(n) =16\omega^{4}\frac{(n+1) (n+\kappa_{0}+\kappa_{1})}{(2n+\kappa_{0}+\kappa_{1})
(2n+\kappa_{0}+\kappa_{1}+1)},
\\
A_{1}^{0}(n,j) =4\frac{(j+1) (j+2) (2n+2\kappa_{0}-1) (2n+2\kappa_{1}-1)
}{(2n+\kappa_{0}+\kappa_{1}-1) (2n+\kappa_{0}+\kappa_{1})}
\\ \hphantom{A_{1}^{0}(n,j) =}
{} \times(4n+2\kappa_{0}+2\kappa_{1}+j-1) (4n+2\kappa_{0}+2\kappa_{1}+j),
\\
A_{0}^{0}(n,j) =-8\omega^{2}(\kappa_{0}-\kappa_{1}) \bigg\{ 2j+\frac{(n+1) j(j-1)
}{2n+\kappa_{0}+\kappa_{1}+1}-\frac{n(j+1) (j+2)}{2n+\kappa_{0}+\kappa_{1}-1}\bigg\}.
\end{gather*}
For $\mathcal{K}\psi_{4n+2+j,j}$, we use $4(n+1) +(j-1) $ and $4n+(j+1) $ for the adjacent labels, and
$E_{4n+2+2j}$, $R\psi_{4n+2+j}=(1+2\kappa_{0}-2\kappa_{1})^{2}%
\psi_{4n+2+j,j}$ and $\mathcal{J}^{2}\psi_{4n+2+j,j}=4(2n+2\kappa
_{0}+1) (2n+2\kappa_{1}+1) \psi_{4n+2+j,j}$ (from
\eqref{J4n2})%
\begin{gather*}
\mathcal{K}\psi_{4n+2+j,j} =A_{-1}^{1}(n) \psi
_{4n+4+j,j-2}+A_{0}^{1}(n,j) \psi_{4n+2+j,j}+A_{1}^{1}(n,j) \psi_{4n+j,j+2},
\\
A_{-1}^{1}(n) =16\omega^{4}\frac{(n+1) (n+\kappa_{0}+\kappa_{1}+1)}{(2n+\kappa_{0}+\kappa_{1}+2)
(2n+\kappa_{0}+\kappa_{1}+1)},
\\
A_{1}^{1}(n,j) =4\frac{(j+1) (j+2) (2n+2\kappa_{0}-1) (2n+2\kappa_{1}+1)
}{(2n+\kappa_{0}+\kappa_{1}+1) (2n+\kappa_{0}+\kappa_{1})},
\\ \hphantom{A_{1}^{1}(n,j) =}
{}\times(4n+2\kappa_{0}+2\kappa_{1}+j+2) (4n+2\kappa_{0}+2\kappa_{1}+j+1)
\\
A_{0}^{1}(n,j) =8\omega^{2}(\kappa_{0}\!+\kappa
_{1}) \!-8\omega^{2}(\kappa_{0}\!-\kappa_{1}\!-1)
\bigg (2j\!+1\!+\frac{(n\!+1) j(j\!-1)
}{2n\!+\kappa_{0}\!+\kappa_{1}\!+2}\!-\frac{n(j\!+1) (j\!+2)
}{2n\!+\kappa_{0}\!+\kappa_{1}}\bigg).
\end{gather*}
For $\mathcal{K}\psi_{j,4n+j}$, use the coefficients from \eqref{cof4nR},
\eqref{cof4n2R}, the reversed labels from $\psi_{4n+j,j}$ and from~\eqref{J4nR}
\[
R\psi_{j,4n+j}=(1-2\kappa_{0}-2\kappa_{1})^{2}\psi
_{j,4n+j},\mathcal{J}^{2}\psi_{j,4n+j}=16n(n+\kappa_{0}+\kappa
_{1}) \psi_{j,4n+j},
\]
then
\begin{gather*}
\mathcal{K}\psi_{j,4n+j} =B_{-1}^{0}(n) \psi_{j-2,4n+j+2}%
+B_{0}^{0}(n,j) \psi_{j,4n+j}+B_{1}^{0}(n,j)\psi_{j+2,4n+j-2},
\\
B_{-1}^{0}(n) =16\omega^{4}\frac{n(n+\kappa
_{0}+\kappa_{1}+1)}{(2n+\kappa_{0}+\kappa_{1}) (2n+\kappa_{0}+\kappa_{1}+1)},
\\
B_{1}^{0}(n,j) =4\frac{(j+1) (j+2) (2n+2\kappa_{0}-1) (2n+2\kappa_{1}-1)
}{(2n+\kappa_{0}+\kappa_{1}-1) (2n+\kappa_{0}+\kappa_{1})}
\\ \hphantom{B_{1}^{0}(n,j) =}
{} \times(4n+2\kappa_{0}+2\kappa_{1}+j-1) (4n+2\kappa_{0}+2\kappa_{1}+j),
\\
B_{0}^{0}(n,j) =-8\omega^{2}(\kappa_{0}-\kappa_{1}) \bigg\{ 2j+2+\frac{nj(j-1)}{2n+\kappa_{0}%
+\kappa_{1}+1}-\frac{(n-1) (j+1) (j+2)}{2n+\kappa_{0}+\kappa_{1}-1}\bigg\}.
\end{gather*}
For $\mathcal{K}\psi_{j,4n+2+j}$, use the reversed labels from $\psi
_{4n+2+j,j}$ and $R\psi_{j,4n++2j}=(1-2\kappa_{0}+2\kappa_{1})
^{2}\psi_{j,4n+j}$, $\mathcal{J}^{2}\psi_{j,4n+2+j}=4(2n+2\kappa
_{0}+1) (2n+2\kappa_{1}+1) \psi_{j,4n+2+j}$ \eqref{J4n2R}%
\begin{gather*}
\mathcal{K}\psi_{j,4n+2+j} =B_{-1}^{1}(n) \psi
_{j-1,4n+3+j}+B_{0}^{1}(n,j) \psi_{j,4n+2+j}+B_{1}^{1}(n,j) \psi_{j+1,4n+j+1},
\\
B_{-1}^{1}(n) =16\omega^{4}\frac{(n+1) (n+\kappa_{0}+\kappa_{1}+1)}{(2n+\kappa_{0}+\kappa_{1}+2)
(2n+\kappa_{0}+\kappa_{1}+1)},
\\
B_{1}^{1}(n,j) =4\frac{(j+1) (j+2) (2n+2\kappa_{0}+1) (2n+2\kappa_{1}-1)
}{(2n+\kappa_{0}+\kappa_{1}+1) (2n+\kappa_{0}+\kappa_{1})}
\\ \hphantom{B_{1}^{1}(n,j) =}
{} \times(4n+2\kappa_{0}+2\kappa_{1}+j+1) (4n+2\kappa_{0}+2\kappa_{1}+j+2),
\\
B_{0}^{1}(n,j) =-8\omega^{2}(\kappa_{0}+\kappa
_{1}) -8\omega^{1}(\kappa_{0}-\kappa_{1}+1)
\\ \hphantom{B_{0}^{1}(n,j) =}
{} \times\bigg(2j+1+\frac{(n+1) j(j-1)
}{2n+\kappa_{0}+\kappa_{1}+2}-\frac{n(j+1) (j+2)
}{2n+\kappa_{0}+\kappa_{1}}\bigg).
\end{gather*}
That concludes the even degree case. One might notice that all the
coefficients are bounded in $(n-j) $, except for the types
$\psi_{2n-j,j}\rightarrow\psi_{2n-j-2,j+2}$, $j\leq n-2$, and $\psi
_{j,2n-j}\rightarrow\psi_{j+2,2n-j-2}$, $j\leq n-3$, which are $O\big((n-j)^{2}\big) $.

\subsection{Odd degree}

The matrix $M$ of $H_{0}-\frac{1}{2}\mathcal{H}$ with respect to the basis
$\{ \psi_{2n+1-j,j}\colon 0\leq j\leq2n+1\} $ has nonzero entries on
the subdiagonal $\{ (i+1,i) \} $, the superdiagonal
$\{ (i,i+1) \} $ (with $0\leq i\leq2n$) and the
cross diagonal $\{ (i,2n+1-i) \colon 0\leq i\leq2n+1\}$.
Because $\big[ H_{0}-\frac{1}{2}\mathcal{H},\sigma_{0}\big] =0$, there is
a symmetry property $M_{i,j}=M_{2n+1-i,2n+1-j}$. Then%
\begin{gather*}
\big(M^{2}\big)_{i,i+2} =M_{i,i+1}M_{i+1,i+2}, \qquad \big(M^{2}\big)_{i,i-2}=M_{i,i-1}M_{i-1,i-2},
\\
\big(M^{2}\big)_{i,i} =M_{i,i+1}M_{i+1,i}+M_{i,i-1}M_{i-1,i}+M_{i,2n+1-i}^{2},
\\
\big(M^{2}\big)_{i,2n-i} =M_{i,i+1}(M_{i+1,2n-i}+M_{i,2n+1-i}),
\\
\big(M^{2}\big)_{i,2n+2-i} =M_{i,i-1}(M_{i-1,2n+2-i}+M_{i,2n+1-i})
\end{gather*}
(because $M_{2n+1-i,2n-i}=M_{i,i+1}$ and $M_{2n+1-i,2n+2-i}=M_{i,i-1}$). The
two cases for $\psi_{2n+1-j,j}$ are $2n+1-2j=1,3\operatorname{mod}4$. The
adjacent ($j\pm1$) polynomials to $\psi_{4n+1+j,j}$ are $\psi_{4n+3+(j-1),j-1}%
$ and $\psi_{4(n-1) +3+(j+1),j+1}$. By use of
\eqref{cof4n1}, \eqref{cof4n3} and Section~\ref{AngMo}, $E_{4n+1+2j}%
,\mathcal{J}^{2}\psi_{4n+1+j,j}=(4n+1+2\kappa_{0}+2\kappa_{1})
^{2}$, $R\psi_{4n+1+j,j}=(1-4\kappa_{0}^{2}-4\kappa_{1}^{2})
\psi_{4n+1+j}$, we obtain
\begin{gather*}
\mathcal{K}\psi_{4n+1+j,j} =A_{-2}(n) \psi_{4n+3+j,j-2}%
+A_{0}(n,j) \psi_{4n+1+j,j}+A_{2}(n,j)\psi_{4n-1+j,j+2}
\\ \hphantom{\mathcal{K}\psi_{4n+1+j,j} =}
{} +A_{1}(n,j) \psi_{j+1,4n+j}+A_{-1}(n,j)\psi_{j-1,4n+2+j},
\\
A_{-2}(n) =16\omega^{4}\frac{(n+1) (n+\kappa_{0}+\kappa_{1}+1)}{(2n+\kappa_{0}+\kappa_{1}+2)
(2n+\kappa_{0}+\kappa_{1}+1)},
\\
A_{2}(n,j) =4\frac{(j+1) (j+2)
(2n+2\kappa_{0}-1) (2n+2\kappa_{1}-1)}{(2n+\kappa_{0}+\kappa_{1}-1) (2n+\kappa_{0}+\kappa_{1})
}
\\ \hphantom{A_{2}(n,j) =}
{}\times(4n+2\kappa_{0}+2\kappa_{1}+j+1) (4n+2\kappa_{0}+2\kappa_{1}+j),
\\
A_{0}(n,j) =-8\omega^{2}\frac{(\kappa_{0}+\kappa
_{1}) (\kappa_{0}-\kappa_{1}) (2n+\kappa
_{0}+\kappa_{1}+j+1)^{2}}{(2n+\kappa_{0}+\kappa_{1})
(2n+\kappa_{0}+\kappa_{1}+1)},
\\
A_{1}(n,j) =-8\omega(\kappa_{0}-\kappa_{1})
\frac{(j+1) (2n+\kappa_{0}+\kappa_{1}+j+1) (4n+2\kappa_{0}+2\kappa_{1}+j+1)}{(2n+\kappa_{0}+\kappa
_{1}-1)_{3}},
\\
A_{-1}(n,j) =-8\omega^{3}(\kappa_{0}+\kappa
_{1}) \frac{(2n+\kappa_{0}+\kappa_{1}+j+1)}{(2n+\kappa_{0}+\kappa_{1})_{3}}.
\end{gather*}
Omit $A_{-2}$ if $j<2$, $A_{2}$ if $n=0$, $A_{1}$ if $n=0$, $A_{-1}$ if $j<1$.
In the special case $\psi_{1,0}=z$, only one term appears: $\mathcal{K}%
z=A_{0}(0,0) z$ and $A_{0}(0,0) =-8\omega
^{2}(\kappa_{0}-\kappa_{1}) (\kappa_{0}+\kappa
_{1}+1) $.

The adjacent ($j\pm1$) polynomials to $\psi_{4n+3+j,j}$ are $\psi_{4(n+1) +1+(j-1),j-1}$ and $\psi_{4n+1+(j+1),j+1}$. By use
of \eqref{cof4n1}, \eqref{cof4n3} and $E_{4n+3+2j}$, $\mathcal{J}^{2}%
\psi_{4n+3+j,j}=(4n+3+2\kappa_{0}+2\kappa_{1})^{2}$,
$R\psi_{4n+3+j,j}=\big({1-4\kappa_{0}^{2}-4\kappa_{1}^{2}}\big)
\psi_{4n+3+j}$, we obtain
\begin{gather*}
\mathcal{K}\psi_{4n+3+j,j} =B_{-2}(n) \psi_{4n+5+j,j-2}%
+B_{0}(n,j) \psi_{4n+3+j,j}+B_{2}(n,j)\psi_{4n+1+j,j+2}
\\ \hphantom{\mathcal{K}\psi_{4n+3+j,j} =}
{} +B_{1}(n,j) \psi_{j+1,4n+2+j}+B_{-1}(n,j)\psi_{j-1,4n+4+j},
\\
B_{-2}(n) =16\omega^{4}\frac{(n+1) (n+\kappa_{0}+\kappa_{1}+1)}{(2n+\kappa_{0}+\kappa_{1}+2)
(2n+\kappa_{0}+\kappa_{1}+3)},
\\
B_{2}(n,j) =4\frac{(j+1) (j+2)
(2n+2\kappa_{0}+1) (2n+2\kappa_{1}+1)}{(2n+\kappa_{0}+\kappa_{1}+1) (2n+\kappa_{0}+\kappa_{1})}
\\ \hphantom{B_{2}(n,j) =}
{} \times(4n+2\kappa_{0}+2\kappa_{1}+j+2) (4n+2\kappa_{0}+2\kappa_{1}+j+3),
\\
B_{0}(n,j) =-8\omega^{2}\frac{(\kappa_{0}+\kappa
_{1}) (\kappa_{0}-\kappa_{1}) (2n+\kappa
_{0}+\kappa_{1}+j+2)^{2}}{(2n+\kappa_{0}+\kappa_{1}+1)
(2n+\kappa_{0}+\kappa_{1}+2)},
\\
B_{1}(n,j) =-8\omega(\kappa_{0}+\kappa_{1})
\frac{(2n+\kappa_{0}+\kappa_{1}+j+2) (4n+2\kappa
_{0}+2\kappa_{1}+j+3)}{(2n+\kappa_{0}+\kappa_{1})_{3}}
\\ \hphantom{B_{1}(n,j) =}
 \times(j+1) (2n+2\kappa_{0}+1) (2n+2\kappa_{1}+1) ,
 \\
B_{-1}(n,j) =-32\omega^{3}(\kappa_{0}-\kappa
_{1}) \frac{(n+1) (2n+\kappa_{0}+\kappa
_{1}+j+1) (n+\kappa_{0}+\kappa_{1}+1)}{(2n+\kappa_{0}+\kappa_{1}+1)_{3}}.
\end{gather*}
Omit $B_{-2}$ if $j<2$, $B_{2}$ if $n=0,B_{-1}$ if $j=0$.

It is perhaps a surprise that the coefficients $A_{0}(n,j) $ and
$B_{0}(n,j) $ are products of linear factors, in contrast to the
even case where the neatest expressions for $A_{0}^{0}$, $A_{0}^{1}$, $B_{0}^{0}$, $B_{0}^{1}$ are partial fractions. In fact, all of the coefficients in this
subsection are products of linear factors, which is not the case for some of
the terms in $H_{0}\psi_{2n+1-j,j}$.

\section{Conclusion}

We described an orthogonal basis of wavefunctions in terms of Jacobi and
Laguerre polynomials. Each of the basis elements is of a particular isotype,
that is, involved in one of the five irreducible representations of the group
$B_{2}$. We defined a fourth order differential-difference self-adjoint
operator $\mathcal{K}$ which commutes with $\mathcal{H}$ but not with the
angular momentum $\mathcal{J}^{2}$.\ This is an example of superintegrability.
The action of $\mathcal{K}$ on the basis elements was found explicitly. It is
known \cite{Quesne2010} that there are differential operators of degree $2k$
which demonstrate superintegrability for the two-parameter $I_{2}(2k)$ (even dihedral group) model. It does not appear straightforward
to adapt the methods of this paper to the larger groups.

\appendix

\section{Transformation of the Hamiltonian}\label{hHh}

This is a short proof of the formula
\[
h_{\kappa}\bigl(-\Delta_{\kappa}+\omega^{2}\Vert x\Vert
^{2}\bigr) h_{\kappa}^{-1}=-\Delta+\omega^{2}\Vert x\Vert
^{2}+\sum_{v\in R_{+}}\frac{\kappa_{v}(\kappa_{v}-\sigma_{v})
\Vert v\Vert^{2}}{\langle x,v\rangle^{2}},
\]
where $h_{\kappa}(x) :=\prod_{v\in R_{+}}\vert
\langle x,v\rangle \vert^{\kappa_{v}}$, the $W(R) $-invariant weight function used in $L^{2}\big(\mathbb{R}^{N},h_{\kappa}^{2}\mathrm{d}m\big) $. For the Laplacian, we have
\begin{align*}
h_{\kappa}\Delta\big(fh_{\kappa}^{-1}\big) -\Delta f&=fh_{\kappa}\Delta
h_{\kappa}^{-1}+2h_{\kappa}\big\langle \nabla f,\nabla h_{\kappa}^{-1}\big\rangle
\\
&=f\sum_{v\in R_{+}}\kappa_{v}\frac{\Vert v\Vert^{2}}{\langle
x,v\rangle^{2}}+f\sum_{i=1}^{N}\bigg(\sum_{v\in R_{+}}\frac
{-\kappa_{v}v_{i}}{\langle x,v\rangle}\bigg)^{2}-2\sum_{v\in
R_{+}}\kappa_{v}\frac{\langle \nabla f,v\rangle}{\langle x,v\rangle}
\end{align*}
and%
\[
\sum_{i=1}^{N}\bigg(\sum_{v\in R_{+}}\frac{-\kappa_{v}v_{i}}{\langle x,v\rangle}\bigg)^{2}
=\sum_{u,v\in R_{+}}\kappa_{u}\kappa_{v}\frac{\langle u,v\rangle}{\langle x,u\rangle
\langle x,v\rangle}=\sum_{v\in R_{+}}\kappa_{v}^{2}\frac
{\Vert v\Vert^{2}}{\langle x,v\rangle^{2}};
\]
this follows from breaking up the double sum over rotations $w=\sigma
_{u}\sigma_{v}$ and the identity $\sigma_{u}^{2}$ and applying a lemma
\cite[Lemma~6.4.6]{DunklXu2014} the $w$ terms vanish. Thus
\[
h_{\kappa}\Delta\big(fh_{\kappa}^{-1}\big) -\Delta f=f\sum_{v\in R_{+}%
}\kappa_{v}(\kappa_{v}+1) \frac{\Vert v\Vert^{2}%
}{\langle x,v\rangle^{2}}-2\sum_{v\in R_{+}}\kappa_{v}%
\frac{\langle \nabla f,v\rangle}{\langle x,v\rangle}.
\]
Also%
\begin{align*}
 \sum_{v\in R_{+}}2\kappa_{v}\frac{h_{\kappa}\big\langle \nabla\big(fh_{\kappa}^{-1}\big),v\big\rangle -\langle \nabla f,v\rangle}{\langle x,v\rangle}
 =-2f\sum_{u,v\in R_{+}}\kappa_{u}\kappa_{v}\frac{\langle
u,v\rangle}{\langle x,u\rangle \langle x,v\rangle
}=-2f\sum_{v\in R_{+}}\kappa_{v}^{2}\frac{\Vert v\Vert^{2}%
}{\langle x,v\rangle^{2}}.
\end{align*}
The other part of $h_{\kappa}\Delta_{\kappa}\big(fh_{\kappa}^{-1}\big) $
contributes $-\sum_{v\in R_{+}}\kappa_{v}\frac{f-\sigma_{v}%
f}{\langle x,v\rangle^{2}}$ thus%
\begin{align*}
h_{\kappa}\Delta_{\kappa}\big(fh_{\kappa}^{-1}\big) -\Delta f &
=\sum_{v\in R_{+}}\kappa_{v}\frac{\Vert v\Vert^{2}}{\langle
x,v\rangle^{2}}\{ (\kappa_{v}+1) f-2\kappa
_{v}f-f+\sigma_{v}f\}
\\
& =-\sum_{v\in R_{+}}\kappa_{v}\frac{\Vert v\Vert^{2}%
}{\langle x,v\rangle^{2}}(\kappa_{v}f-\sigma_{v}f).
\end{align*}
This proves the formula.

\section{Symbolic computation proofs}\label{Symb}

There is an analog $K(x,y) $ of the exponential function
$\exp\langle x,y\rangle $ on $\mathbb{R}^{N}\times\mathbb{R}^{N}$
which satisfies $K(x,y) =K(y,x)$, $K(xw,yw) =K(x,y) $ for all $w\in W(R) $ and
$\mathcal{D}_{i}^{(x)}K(x,y) =y_{i}K(x,y) $ (where~$\mathcal{D}_{i}^{(x)}$ is the operator
$\mathcal{D}_{i}$ acting on $x$, for $1\leq i\leq N$). The kernel exists for
nonsingular parameters $\{ \kappa_{v}\} $, which include the
situation $\kappa_{v}\geq0$. Suppose $p(x) $ is a polynomial
then by the product rule%
\begin{align*}
\mathcal{D}_{i}(p(x) K(x,y) )
=\bigg(y_{i}p(x) +\frac{\partial}{\partial x_{i}}p(x) \bigg) K(x,y)
 +\sum_{v\in R_{+}}\kappa_{v}\frac{p(x) -p(x\sigma
_{v})}{\langle x,v\rangle}K(x\sigma_{v},y)
v_{i}.
\end{align*}
This formula together with $wK(x,y) =K(xw,y)
=K\big(x,yw^{-1}\big) $ show how an element of the rational Cherednik
algebra (an algebra of operators on polynomials generated by $\big\{
\mathcal{D}_{i}^{(x)},x_{i}\colon 1\leq i\leq N\big\} \cup
W(R) $) acts on a generic sum $\sum_{w\in W(R)
}p_{W}(x,y) K(xw,y) $. It can be shown that if~$\mathcal{T}$ is in the rational Cherednik algebra and $\mathcal{T}K(x,y) =0$, then $\mathcal{T}=0$ (see Dunkl~\cite{Dunkl1999}). For
particular groups and operators, the calculation of $\mathcal{T}K(x,y) $ can be implemented in computer algebra. The function $K$ is an
undefined function with argument $\langle x,y\rangle $ \big(or
$\big\langle x,yw^{-1}\big\rangle$\big). To~compute $\mathcal{D}_{i}^{(x)}K( xw,y) =\mathcal{D}_{i}^{(x)}K\big(x,yw^{-1}\big) =\big(yw^{-1}\big)_{i}K( xw,y) $ one
applies $\frac{\partial}{\partial x_{i}}$ to $\big\langle x,yw^{-1}%
\big\rangle$, a~straightforward calculation.

In the $B_{2}$ application with complex coordinates $z=x_{1}+\mathrm{i}%
x_{2}$, $u=y_{1}+\mathrm{i}y_{2}$, the inner product is $\langle
x,y\rangle =\frac{1}{2}(z\overline{u}+\overline{z}u) $. As
examples,
\begin{gather*}
\overline{T}K\bigg(\frac{1}{2}(z\overline{u}+\overline{z}u)\bigg) =\frac{1}{2}uK\bigg(\frac{1}{2}(z\overline{u}+\overline{z}u) \bigg),
\\
T\bigg\{ \big(z^{2}-\overline{z}^{2}\big) K\bigg(\frac{1}{2}(z\overline{u}+\overline{z}u)\bigg)\bigg\}
=\bigg\{\frac{1}{2}\big(z^{2}-\overline{z}^{2}\big) \overline{u}+2z\bigg\} K\bigg(\frac{1}{2}( z\overline{u}+\overline{z}u) \bigg)
\\ \hphantom{T\bigg\{ \big(z^{2}-\overline{z}^{2}\big) K\bigg(\frac{1}{2}(z\overline{u}+\overline{z}u)\bigg)\bigg\}=}
{}+2\kappa_{0}(z-\overline{z}) K\bigg({-}\frac{1}{2}(zu+\overline{z}\overline{u}) \bigg)
\\ \hphantom{T\bigg\{ \big(z^{2}-\overline{z}^{2}\big) K\bigg(\frac{1}{2}(z\overline{u}+\overline{z}u)\bigg)\bigg\}=}
{}+2\kappa_{0}(z+\overline{z}) K\bigg(\frac{1}{2}(zu+\overline{z}\overline{u}) \bigg),
\end{gather*}
and%
\begin{align*}
K\left(-\frac{1}{2}(zu+\overline{z}\overline{u}) \right)
=\sigma_{2}K\left(\frac{1}{2}(z\overline{u}+\overline{z}u)\right),\qquad
K\left(\frac{1}{2}(zu+\overline{z}\overline{u}) \right)
=\sigma_{0}K\left(\frac{1}{2}(z\overline{u}+\overline{z}u)\right).
\end{align*}
This method is used to prove Theorem~\ref{H0123sq}.

\subsection*{Acknowledgements}

The author presented some of the material in a plenary lecture at the 34th
International Colloquium on Group Theoretical Methods in Physics, Strasbourg,
France, July 18--22, 2022. Also the author thanks the referees for their careful
reading and constructive suggestions.

\pdfbookmark[1]{References}{ref}
\LastPageEnding

\end{document}